\pdfoutput=1
\documentclass{scrartcl}

\usepackage[utf8]{inputenc}
\usepackage[T1]{fontenc}
\usepackage{amsmath}
\usepackage{amssymb}
\usepackage{amsthm}
\usepackage{hyperref}
\usepackage{enumerate}
\usepackage{lmodern}
\usepackage{url}
\usepackage{chemformula}
\usepackage{pgfplots}
\usepackage{pgf-pie}
\usepackage{xspace}
\usepackage{longtable}

\usepackage{algorithm,algorithmic}

\newtheorem{theorem}{Theorem}
\newtheorem{proposition}[theorem]{Proposition}
\newtheorem{lemma}[theorem]{Lemma}
\theoremstyle{definition}
\newtheorem{definition}[theorem]{Definition}
\theoremstyle{remark}
\newtheorem{remark}[theorem]{Remark} 
\newtheorem{example}[theorem]{Example}
\theoremstyle{corollary}
\newtheorem{corollary}[theorem]{Corollary}

\newcommand{\NN}{{\mathbb N}}
\newcommand{\ZZ}{{\mathbb Z}}
\newcommand{\QQ}{{\mathbb Q}}
\newcommand{\RR}{{\mathbb R}}
\newcommand{\RRn}{\RR^n}
\newcommand{\Rplus}{\RR_{>0}}
\newcommand{\Rplusn}{(\Rplus)^n}
\newcommand{\CC}{{\mathbb C}}
\newcommand{\CCn}{\CC^n}
\newcommand{\Cstar}{\CC^*}
\newcommand{\Cstarn}{(\Cstar)^n}
\newcommand{\Cstarm}{(\Cstar)^m}
\newcommand{\KK}{K}
\newcommand{\KKn}{\KK^n}
\newcommand{\Kstar}{K^*}
\newcommand{\Kstarn}{(\Kstar)^n}
\newcommand{\Kstarm}{(\Kstar)^m}
\newcommand{\VV}{V}
\newcommand{\Vstar}{V^*}
\newcommand{\laurentseries}{\KK[x_1^{\pm 1},\ldots, x_n^{\pm 1}]}
\newcommand{\xgamma}{X^\gamma}
\newcommand{\xalpha}{X^\alpha}
\newcommand{\xbeta}{X^\beta}
\newcommand{\xalphaminusxbeta}{\xalpha-\xbeta}

\newcommand{\ideal}[1]{\langle #1 \rangle}
\newcommand{\Zx}{\ZZ[x_1,\ldots,x_n]}
\newcommand{\Zxpm}{\ZZ[x_1^{\pm},\ldots,x_n^{\pm}]}
\newcommand{\tGG}{\mathtt{G}}
\newcommand{\tCC}{\mathtt{C}}
\newcommand{\tXX}{\mathtt{X}}
\newcommand{\tOO}{\mathtt{O}}
\newcommand{\tgg}{\mathtt{g}}
\newcommand{\tcc}{\mathtt{c}}
\newcommand{\txx}{\mathtt{x}}
\newcommand{\too}{\mathtt{o}}
\newcommand{\false}{\operatorname{false}}
\newcommand{\true}{\operatorname{true}}
\newcommand{\dotequal}{\mathrel{\dot{=}}}
\newcommand{\vars}{\operatorname{vars}}
\newcommand{\Classify}{\operatorname{Classify}}
\newcommand{\ProjectAndClassify}{\operatorname{ProjectAndClassify}}
\newcommand{\DecomposeAndClassify}{\operatorname{DecomposeProjectAndClassify}}
\newcommand{\GB}{\operatorname{GroebnerBasis}}
\newcommand{\LOR}{L_\text{OR}}
\renewcommand{\emptyset}{\varnothing}
\def\nmodels{\mathrel|\joinrel\neq}

\newcommand{\gb}{Gr\"obner basis\xspace}
\newcommand{\gbs}{Gr\"obner bases\xspace}
\newcommand{\crn}{chemical reaction network\xspace}
\newcommand{\crns}{chemical reaction networks\xspace}

\allowdisplaybreaks

\KOMAoptions{
  abstract=true,
  fontsize=10pt,
  DIV=10,
  headinclude=false,
  footinclude=false}

\addtokomafont{author}{\large}

\recalctypearea

\begin{document}
\title{Efficiently and Effectively Recognizing Toricity of Steady State
  Varieties}

\author{
  Dima Grigoriev, CNRS and the University of Lille, France\\
  \texttt{dmitry.grigoryev@math.univ-lille1.fr}
  \and
  Alexandru Iosif, RWTH Aachen University, Germany\\
  \texttt{iosif@aices.rwth-aachen.de}
  \and
  Hamid Rahkooy, CNRS, Inria, and the University of Lorraine, France\\
  \texttt{hamid.rahkooy@inria.fr}
  \and
  Thomas Sturm,
  CNRS, Inria, and the University of Lorraine, France\\
  MPI Informatics and Saarland University, Germany\\
  \texttt{thomas.sturm@loria.fr}
  \and
  Andreas Weber, University of Bonn, Germany\\
  \texttt{weber@informatik.uni-bonn.de}}

\date{April 15, 2020}

\maketitle

\begin{abstract}
  We consider the problem of testing whether the points in a complex or real
  variety with non-zero coordinates form a multiplicative group or, more
  generally, a coset of a multiplicative group. For the coset case, we study the
  notion of shifted toric varieties which generalizes the notion of toric
  varieties. This requires a geometric view on the varieties rather than an
  algebraic view on the ideals. We present algorithms and computations on 129
  models from the BioModels repository testing for group and coset structures
  over both the complex numbers and the real numbers. Our methods over the
  complex numbers are based on Gröbner basis techniques and binomiality tests.
  Over the real numbers we use first-order characterizations and employ real
  quantifier elimination. In combination with suitable prime decompositions and
  restrictions to subspaces it turns out that almost all models show coset
  structure. Beyond our practical computations, we give upper bounds on the
  asymptotic worst-case complexity of the corresponding problems by proposing
  single exponential algorithms that test complex or real varieties for toricity
  or shifted toricity. In the positive case, these algorithms produce generating
  binomials. In addition, we propose an asymptotically fast algorithm for
  testing membership in a binomial variety over the algebraic closure of the
  rational numbers.
\end{abstract}

\section{Introduction}
We are interested in situations where the points with non-zero coordinates in a
given complex or real variety form a multiplicative group or, more generally, a
coset. We illustrate this by means of a simple example. For $K\in\{\CC,\RR\}$,
$(K^*)^n$ denotes the direct power of the multiplicative group of the respective
field. Consider a family of ideals
\begin{equation}\label{eq:ideal_crn}
  I_k=\ideal{x^2-ky^2}
\end{equation}
with a rational parameter $k$. Let $V_k$ be the complex variety of $I_k$, and
let $\Vstar_k=\VV_k \cap (\Cstar)^2$. Then $\Vstar_1$ forms a \emph{group}, but
$\Vstar_{-1}$ does not, because it does not contain $(1,1)$. However,
$\Vstar_{-1}=(1,i) \cdot \Vstar_1$ forms a \emph{coset} of $\Vstar_1$. Over the
reals, $V_1^*\cap(\RR^*)^2$ is again a group, but
$V_{-1}^*\cap(\RR^*)^2 = \emptyset$ is not a coset of any group. Consider now
$\VV_1 = V_{11} \cup V_{12}$, where $V_{11}$ and $V_{12}$ are given by
$\ideal{x-y}$ and $\ideal{x+y}$, respectively. Notice that $\Vstar_{11}$ is
itself a group, and $\Vstar_{12}$ is a coset of $\Vstar_{11}$. Both
$\Vstar_{11}$ and $\Vstar_{12}$ are irreducible because their generating ideals
are prime. Under the additional condition of irreducibility $\Vstar_{11}$ forms
a \emph{torus} and $\Vstar_{12}$ forms a \emph{shifted torus}. We then call the
varieties $V_{11}$ and $V_{12}$ \emph{toric} and \emph{shifted toric},
respectively.

Toric varieties are well established and have an important role in algebraic
geometry \cite{fulton_introduction_2016,eisenbud-sturmfels-binomials}. However,
our principal motivation to study generalizations of toricity comes from the
sciences, specifically \textit{\crns} such as the following:
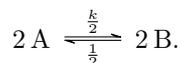
\begin{center}
  \ch[label-style=\scriptsize]{2 A <=>[$\frac{k}{2}$][$\frac12$] 2 B}.
\end{center}
Assuming, e.g., mass action kinetics \cite{voit2015150} one can derive a system
of autonomous ordinary differential equations describing the development of
concentrations of the species A and B as functions of time \cite[Section
2.1.2]{feinberg-book}. For the given reaction network one obtains a polynomial
vector field generating exactly our ideals $I_k$ in (\ref{eq:ideal_crn}). Our
methods thus detect whether equilibrium points with non-zero coordinates form a
group or a coset.

Detecting toricity of a variety in general, and of the steady state varieties of
chemical reaction networks in particular, is a difficult problem
\cite{sturm-parametric-steady-state-issac19}. The first issue in this regard is
finding suitable notions to describe the structure of the steady states.
Existing work typically addresses algebraic properties of the steady state
ideal, e.g., the existence of binomial Gröbner bases. In this article, in
contrast, we take a geometric approach, focusing on varieties rather than
ideals. We propose to study toricity and shifted toricity of varieties $V$ over
$K\in\{\CC,\RR\}$, which for irreducible varieties coincides with $V\cap(K^*)^n$ forming
a multiplicative group or coset, respectively. It is noteworthy that chemical
reaction network theory generally takes place in the interior of the first
orthant of $\RR^n$, i.e., all species concentrations and reaction rates are
assumed to be strictly positive \cite{feinberg-book}. Our considering
$(\CC^*)^n$ in contrast to $\CC^n$ is a first step in this direction, considering
also $(\RR^*)^n$ is another step.

Toric dynamical systems have been studied by Feinberg \cite{Feinberg1972} and by
Horn and Jackson \cite{Horn1972}. Craciun et al.~\cite{craciun_toric_2009}
showed that toric dynamical systems correspond to \textit{complex balancing}
\cite{feinberg-book}. Our generalized notions of toricity are inspired by
Grigoriev and Milman's work on \textit{binomial varieties}
\cite{grigoriev_milman2012}. There are further definitions in the literature,
where the use of the term ``toric'' is well motivated. Gatermann et
al.~considered \textit{deformed toricity} for steady state ideals
\cite{Gatermann2000}. The exact relation between the principle of complex
balancing and various definitions of toricity has obtained considerable
attention in the last years
\cite{perez_millan_chemical_2012,gatermann_bernsteins_2005,muller_sign_2016}.

Complex balancing itself generalizes \textit{detailed balancing}, which has
widely been used in the context of chemical reaction networks
\cite{feinberg_stability_1987,feinberg-book,Horn1972}. Gorban et
al.~\cite{GorbanYablonski:11a,GorbanMirkes:13a} related reversibility of
chemical reactions in detailed balance to binomiality of the corresponding
varieties. Historically, the principle of detailed balancing has attracted
considerable attention in the sciences. It was used by Boltzmann in 1872 in
order to prove his \textit{H-theorem} \cite{boltzmann1964lectures}, by Einstein
in 1916 for his quantum theory of emission and absorption of radiation
\cite{einstein1916strahlungs}, and by Wegscheider \cite{Wegscheider1901} and
Onsager \cite{onsager1931reciprocal} in the context of \textit{chemical
  kinetics}, which lead to Onsager's Nobel prize in Chemistry in 1968.
Pérez-Millán et al.~\cite{perez_millan_chemical_2012} consider steady state
ideals with binomial generators. They present a sufficient linear algebra
condition on the \textit{stoichiometry matrix} of a \crn in order to test
whether the steady state ideal has binomial generators. Conradi and Kahle
proposed a corresponding heuristic algorithm. They furthermore showed that the
sufficient condition is even equivalent when the ideal is homogenous
\cite{conradi2015detecting,kahle-binomial-package-2010,Kahle2010}. Based on the
above-mentioned linear algebra condition, MESSI systems have been introduced in
\cite{millan_structure_2018}. Recently, binomiality of steady state ideals was
used to infer network structure of chemical reaction networks out of measurement
data \cite{Wang_Lin_Sontag_Sorger_2019}.

Besides its scientific adequacy as a generalization of complex balancing there
are practical motivations for studying toricity. Relevant models are typically
quite large. For instance, with our comprehensive computations in this article
we will encounter systems up to 90 polynomials in dimension 71. This potentially
takes symbolic computation to its limits. A possible approach to overcome this
is to discover systematic occurrences of certain structural properties in the
input, and to exploit those structural properties towards more efficient
algorithms. From this point of view, toricity and shifted toricity are
interesting concepts because tools from toric geometry can be used as a
complexity reduction step in the multistationarity problem. Grigoriev and Weber
\cite{GrigorievWeber2012a} gave a complexity analysis for solving binomial
varieties, based on the computation of Smith normal forms. More interestingly,
toric and shifted toric models are known to have scale invariant
multistationarity in the space of linear conserved quantities, which further
reduces the dimension of the multistationarity problem \cite{Conradi2019}.

Our original contributions in this article are the following. Interested in the
geometric structure of real and complex varieties $V$ rather than the algebraic
structure of the corresponding ideals, we study primarily $V^*=V\cap(K^*)^n$. We
call $V$ shifted toric when $V$ is irreducible and $V^*$ is a coset. This
generalizes the notion of toric varieties $V$ for groups
$V^*$. Within this setting, we have two principal results:
\begin{itemize}
\item Relating our novel geometric view to the established algebraic view, we
  give a characterization in terms of Gröbner bases for $V^*$ to be a group or
  coset. (Proposition~\ref{prop:group-coset-iff}).
\item We show that Zariski closures of groups in $(\CC^*)^n$ are binomial
  varieties (Proposition~\ref{prop:groupstructure}). The converse had been
  shown in \cite{grigoriev_milman2012}.
\end{itemize}

We propose practical algorithms testing for given polynomial systems $F$
whether their varieties contain group or coset structures.
\begin{itemize}
\item We consider over the complex numbers $V(F)^*$ (Algorithm~\ref{alg:pacc})
  and $V(P_i)^*$ for prime components $\langle P_i\rangle$ of $\langle F\rangle$
  (Algorithm~\ref{alg:dacc}).
\item We consider the same over the real numbers, $V(F)^*$
  (Algorithm~\ref{alg:pacr}) and $V(P_i)^*$ (Algorithm~\ref{alg:dacr}).
\item With prime decomposition we find that for up to 98\% of the prime
  components $V_K(P_i)^*$ is either empty or a coset.
\end{itemize}
Our algorithms are implemented in Maple\footnote{Maple (2019). Maplesoft, a
  division of Waterloo Maple Inc., Waterloo, Ontario.} and
Reduce\footnote{\url{https://sourceforge.net/projects/reduce-algebra/}}
\cite{Hearn:05a,Hearn:67a} and systematically applied to the steady state
varieties of 129 models from the BioModels
repository\footnote{\url{https://www.ebi.ac.uk/biomodels/}}. Our objective was
to build on robust, off-the-shelf software, which has a chance to be accepted by
scientific communities outside symbolic computation in the foreseeable future.
As a consequence, our proposed algorithms must rely on existing implementations.
Over $\CC$ we use Gröbner bases \cite{Buchberger:65a,Faugere:99a,Faugere:02a}, and
over $\RR$ we use real quantifier elimination techniques.

Gröbner bases and real quantifier elimination mentioned above come with high
intrinsic complexity. The former are complete for exponential space
\cite{mayr_complexity_1982}. The latter are double exponential
\cite{Davenport:1988:RQE:53372.53374,Grigoriev:88a,Weispfenning:88a}. From a
more theoretical point of view we study the intrinsic complexity of the problems
actually addressed. We follow Chistov--Grigoriev's complexity results for
first-order quantifier elimination over algebraically closed fields
\cite{chistov_complexity_1984}, the algorithm constructing irreducible
components of a variety \cite{chistov_complexity_1984,
  grigoriev_factorization_1986} and Grigoriev--Vorobjov's algorithm for solving
polynomial system of inequalities \cite{grigoriev_solving_1988}. These results
are used to propose an algorithm to test within single exponential complexity
bounds whether:
\begin{itemize}
\item a given complex variety is toric or shifted toric
  (Theorem~\ref{complex});
\item a given real variety is toric or shifted toric
  (Theorem~\ref{real});
\item a given point belongs to a given binomial variety
  (Theorem~\ref{th:membership}).
\end{itemize}

The plan of the article is as follows. In Section~\ref{sec:preliminaries} we
present preliminaries from the literature and introduce our new notions and
related results. In Section \ref{sec:computations}, we present new algorithms
for group and coset tests over $\CC$ and $\RR$. As a first step towards irreducible
varieties we also use prime decompositions over the coefficient field, i.e.,
rational numbers. For the sake of a concise discussion, the major part of our
rather comprehensive computation results can be found in Appendix~\ref{app:a}.
In Section \ref{sec:complexities} we propose asymptotically fast algorithms for
the practical computations in Section~\ref{sec:computations}. The proposed
algorithms induce upper complexity bounds on the corresponding problems. In
Section \ref{sec:conclusions} we summarize our findings and mention perspectives
for future work.

\section{Toric, Shifted Toric and Binomial
  Varieties}\label{sec:preliminaries} 
We use $K$ to denote either $\CC$ or $\RR$ when definitions or results hold for both
fields. The natural numbers $\NN$ include $0$. For positive $n\in\NN$ and
$X=(x_1,\dots,x_n)$, the polynomial ring with coefficients in $\QQ$ and variables
$x_1$, \dots,~$x_n$ is written $\QQ[X]=\QQ[x_1,\dots,x_n]$. For
$\alpha=(\alpha_1,\dots,\alpha_n) \in \NN^n, X^\alpha=x_1^{\alpha_1}\dots x_n^{\alpha_n}$ is a monomial in
$\QQ[X]$. When mentioning Gröbner bases of ideals we always mean reduced Gröbner
bases; when not mentioned explicitly the term order is not relevant. Given a
polynomial $f \in \QQ[X]$, the variety of a $f$ over $K$ is
$V(f)=\{\,z\in K\mid f(z)=0\,\}\subseteq K^n$; this naturally generalizes to sets
$F$ of polynomials and ideals $I$. Vice versa, given a variety $V\subseteq K^n$, we
define the ideal of $V$ to be $I(V)=\{\,f\in\QQ[X]\mid\text{$f(z)=0$ for all
  $z\in V$}\,\}$. Recall that over algebraically closed fields,
$I(V(J)) = \sqrt{J}$, the radical of $J$.

Let $\Kstar$ be the multiplicative group of $\KK$. A subgroup
$G \subseteq \Kstarn$ is called a \textit{torus} over the direct product
$\Kstarn$, where multiplication is coordinate-wise, if there exists $m \in \NN$ such
that $G$ is isomorphic to $\Kstarm$. A variety $V \subseteq \KKn$ is called
\textit{toric} if it is irreducible and there exists a torus
$G \subseteq \Kstarn$ such that $V=\overline{G}$, the Zariski closure of $G$
\cite{sturmfels_grobner_1996}. It is noteworthy that there are alternative
definitions that requires the variety to be normal as well
\cite{fulton_introduction_2016}.

For a variety $V \subseteq \KKn$, by $\Vstar$ we denote $V \cap \Kstarn$, i.e., the points
in $V$ with non-zero coordinates. It is well-known that every torus is the
(irreducible) zero set of a set of Laurent binomials of the form $\xgamma-1$
where $\gamma \in \ZZ$ \cite{fulton_introduction_2016}. We are going to make use of the
following proposition, which is a consequence of results in
\cite{fulton_introduction_2016,sturmfels_grobner_1996}.

\begin{proposition}\label{idealtoricvariety}
  Let $V \subseteq \CCn$ be a variety. $V$ is a toric variety if and only $I(V)$ is
  prime and the reduced \gb (with respect to any term order) of $I(V)$ contains
  only binomials of the form $\xalphaminusxbeta$ where $\alpha$, $\beta \in \NN^n$.
\end{proposition}

By definition, $V$ is toric if and only if $V$ is irreducible and there exists a
torus $T$ such that $V=\overline T$. Assume that $V$ is given by a set of
generators of $I(V)$. Since $V$ is irreducible, then $\Vstar$ is irreducible,
hence $T=\Vstar$, $V = \overline \Vstar$, and $I(V)=I(\overline \Vstar)$.
Therefore, it suffices to compute a \gb of $I(V)$ instead of $I(\Vstar)$ and use
Proposition \ref{idealtoricvariety} in order to check if $\Vstar$ is a group.
Note that if there is a \gb of $I(V)$ where all elements are of the form
$X^\alpha-X^\beta$ and the ideal is prime, then $\Vstar$ is clearly non-empty. For more
detailed study of toric varieties refer to \cite{fulton_introduction_2016} and
\cite{sturmfels_grobner_1996}.

Not all subgroups of $\Kstarn$ are reducible. For example if $K=\CC$ and
$V=V(x^2-y^2) \subseteq \CCn$, one can check that $\Vstar = V \cap (\Cstar)^2$
is a group, however it is not reducible, hence not a torus. Actually $V$ can be
decomposed into the torus $V_1=V(x-y)\cap (\Cstar)^2$ and
$V_2=V(x+y) \cap (\Cstar)^2$, which is a coset of $\VV_1$. Varieties that admit
a group structure have interesting properties. A class of such varieties, called
binomial varieties are studied by Grigoriev and Milman
\cite{grigoriev_milman2012}, where the authors present a structure theorem for
them and discuss the complexity of their Nash resolution. We remind ourselves of
the definition and the structure theorem of binomial varieties, which will be
used for classifying steady state ideals.
\begin{definition}[Binomial Variety]
  A variety $V \subseteq \CCn$ is called binomial if
  $\Vstar := V \cap \Cstarn$ is the zero set of a finite set of
  binomials of the form $\xalpha - 1 \in \laurentseries$ and
  $V = \overline{\Vstar}$.
\end{definition}

Using \gbs instead of Laurent polynomials, one can see that if a
variety $V \subseteq \CCn$ is binomial then all elements of every \gb
of $I(V)$ are binomials of the form $\xalphaminusxbeta$, where
$\alpha,\beta \in \NN\setminus \{0\}$. The following theorem by
Grigoriev and Milman shows the structure of the binomial varieties by
precisely describing their irreducible components.

\begin{theorem}\cite[Theorem
  3.7]{grigoriev_milman2012}\label{structurebinomialvariety}
  The irreducible components of a binomial variety $V \subseteq \CCn$ include
  exactly one toric variety $\overline{T}$, where $T \subseteq \Cstarn$ is a
  torus, and several varieties $V_1=\overline{x_1T}$,
  \dots,~$V_r=\overline{x_rT}$, where $x_1T$, \ldots,~$x_rT \subseteq \Cstarn$
  are cosets of $T$ as a group, with respect to $x_1$, \dots,~$x_r \in \Cstarn$,
  respectively.
\end{theorem}

Later in this section we will show that the closure of any subgroup of
$\Cstarn$ is a binomial variety. Proposition \ref{idealtoricvariety}
gives the form of the polynomials in a \gb of the ideal describing the
toric component of a binomial variety.  For the components that are
cosets of the torus, one can easily derive the form of the reduced \gb
from the definition of the torus and Proposition
\ref{idealtoricvariety}. This is stated precisely in the following
proposition.  Intuitively, non-toric components of a binomial variety
can be considered as the \textit{shifts} of the toric component. This
motivates us to define \textit{shifted toric varieties}.

\begin{definition}[Shifted Toric Variety]\label{def:shiftedtoricvariety}
  A shifted torus in $\Kstarn$ is defined to be a coset of a torus in
  $\Kstarn$. A variety $\VV \subseteq \KKn$ is called shifted toric if
  it is the closure of a shifted torus.
\end{definition}

Since every group is a coset of itself, every torus is a shifted torus, and
therefore every toric variety is a shifted toric variety. However, a shifted
toric variety is not necessarily a toric variety. Following the definition of
shifted toric varieties and using Proposition \ref{idealtoricvariety}, we show
in the following proposition that ideals of shifted toric varieties have \gbs of
a specific form.

\begin{proposition}\label{prop:gb-shiftedtoric}
  $\Vstar \subseteq \CCn$ is a shifted torus if and only if $\Vstar$ is a zero set of
  Laurent binomials of the form $\left(\frac{X}{g}\right)^{\alpha} -1$, where
  $g \in \Vstar$ and $\alpha \in \ZZ^n$. Similarly, $V$ is shifted toric if and only if
  $I(V)$ is prime and the reduced \gb (with respect to any term order) of $I(V)$
  contains only binomials of the form $\xalpha + c \xbeta$ where $c \in \Cstar$,
  $\alpha$, $\beta \in \NN^n\setminus \{(0,\dots,0)\}$.
\end{proposition}
\begin{proof}
  $V$ is shifted toric if and only if there exists a torus $T$ and a coset $C$
  of $T$ with respect to some $g \in \KKn$ (i.e., $C=gT$), such that
  $V = \overline C$. Since $V$ is irreducible, then
  $\overline \Vstar = \overline C =V$ and
  $I(V)=I(\overline \Vstar) = I(\overline C)$, and this ideal is prime. Assume
  that $V$ given by a set of generators of $I(V)$. $V$ is shifted toric if and
  only if $I(\overline \Vstar)=I(V)$ is prime and $\Vstar$ is a coset. This
  holds if and only if $I(\Vstar)$ is prime and $g^{-1}\Vstar$ is a group. Note
  that $\overline {g^{-1}\Vstar}$ is irreducible if and only if
  $\overline \Vstar$ is irreducible. This holds if and only if $I(\Vstar)$ is
  prime, or equivalently $I(V)$ is prime. Therefore, by Proposition
  \ref{idealtoricvariety}, $V$ is shifted toric if and only if $I(V)$ is prime
  and all the elements of every \gb of $I(\overline {g^{-1}\Vstar})$ is of the
  form $X^\alpha -X^\beta$, for
  $\alpha,\beta \in \NN^n \setminus \{(0,\dots,0)\}$. Equivalently, all the elements of
  $I(V)$ are of the form
  $\left(\frac{X}{g}\right)^\alpha-\left(\frac{X}{g}\right)^\beta$. Cleaning the
  denominator, we have the desired form of the \gb elements.
\end{proof}

Proposition \ref{prop:gb-shiftedtoric} along with the structure theorem for
binomial ideals imply that the primary decomposition of an ideal generated by
binomials of the form $X^\alpha-X^\beta$ include an ideal generated by binomials of the
form $X^\alpha-X^\beta$ and several ideals generated by binomials of the form
$X^\alpha-c X^\beta$.

Using Proposition \ref{prop:gb-shiftedtoric}, one can design a randomised
algorithm for testing shifted toricity of a variety $V = \overline{\Vstar}$. Let
$g_1,\dots,g_m$ be \textit{generic points} in $\Vstar$ and consider the set of
Laurent binomials
$ G=\left\{\,\left( \frac{X}{g_i} \right) ^{\gamma_i} -1\ \middle\vert\ i=1,\dots,m\,
\right\}$ with symbolic exponents $\gamma_i =(\gamma_{i1},\ldots, \gamma_{in})$. Let
\begin{equation}
M=\left(\begin{array}{c c c}
           \gamma_{11} &  \cdots & \gamma_{1n} \\
           \vdots & \ddots & \vdots \\
           \gamma_{m1} & \cdots & \gamma_{mn} \\
        \end{array} \right)
\end{equation}
be the matrix of exponents of the Laurent binomials and make it row
reduced. This leads to a linear combination of rows with coefficients
in $\ZZ$, say $d_i$.  Then $\VV$ is shifted toric if and only if
\begin{equation}\label{eq:rand-shifted}
  \sum\limits_{i=1}^m d_i\gamma_i = 0,
\end{equation}
which holds if and only if
$\prod\limits_{i=1}^m \left( g_{i1}^{d_{i1}}\ldots g_{in}^{d_{in}} \right) =1$.
Solving the linear equations (\ref{eq:rand-shifted}) will give us the exponents
of the Laurent polynomials. Note that $G$ obtained in this way is a reduced \gb.

One can see that a binomial variety is the closure of a group and furthermore,
by Proposition \ref{structurebinomialvariety}, it can be decomposed into toric
and shifted toric varieties as its irreducible components. A natural question is
whether this property holds for every variety that is the closure of a group.
The answer to this question is positive and indeed such varieties are precisely
binomial varieties. This is explicitly formulated in a remark in \cite[after
Proposition 2.3]{eisenbud-sturmfels-binomials}.

\begin{proposition}\label{prop:groupstructure}
  Let $W$ be a subgroup of $\Cstarn$.  Then $\overline W$ is a
  binomial variety.
\end{proposition}
\begin{proof}
  By definition of binomial variety, we have to prove that
  $\overline W$ is the zero set of binomials of the form
  $X^\alpha-X^\beta$ and
  $\overline{\left(\overline{W}\right) ^*}=\overline W$.  The equality
  $\overline{\left(\overline{W}\right) ^*}=\overline W$ directly comes
  from the definition of Zariski closure.

  For proving that the generators of $I(\overline W)$ have the desired
  form, we use the notations of
  \cite{eisenbud-sturmfels-binomials}. By Proposition 2.3(a) in the
  latter reference, $\CC[X^{\pm 1}]I(W)=I(\rho)$, for some partial
  character $\rho\in\text{Hom}(\ZZ^{n},\CC^{*})$. By Theorem 2.1(b) in
  the same reference,
  $I(\rho)=\langle x^{m_{1}}-\rho(m_{1}),\ldots,
  x^{m_{r}}-\rho(m_{r})\rangle$ where $m_{1},\ldots,m_{r} \in\ZZ^n$ is
  a basis of $L_{\rho}$. As $W$ is a group, $(1,\dots,1)\in W$; hence
  $\rho(m_{1})=\ldots=\rho(m_{r})=1$ and therefore $I(W)$ is generated
  by binomials of the form $X^{\alpha}-X^{\beta}$ where
  $\alpha,\beta \in \NN$. Since $\overline W = V(I(W))$, we have
  proved the proposition.
\end{proof}

Proposition \ref{prop:groupstructure} can also be proved over the positive real
numbers by considering the logarithm map on $\Rplusn$ acting coordinate-wise.
The image of this map forms a linear space. A basis of this linear space
provides a parametrization of a group.

From a computational point of view, a variety $V=V(I)$ is usually given by a set
of generators of $I$ and we would like to derive information about toricity,
binomiality or coset property of $V$ by computations over the generators of $I$.
This can be done via \gbs. Assume that $G$ is a \gb of $I$, hence $V=V(G)$, and
$\Vstar \ne \emptyset$. If all elements of $G$ are of the form
$X^\alpha-X^\beta$, then $\Vstar$ is a subgroup of $\Cstarn$. If all elements of
$G$ are of the form $c_\alpha X^\alpha-c_\beta X^\beta$ where $c_\alpha\ne 0$ and
$c_\beta \ne 0$, then $\Vstar$ is a coset of a subgroup of $\Cstarn$. Note that the
converse of the above does not hold. This is because $\overline \Vstar$ and $V$
may not be equal and therefore $I(\overline\Vstar)$ and $I(V)$ may not be equal,
which means that a \gb of $I(V)$ does not give information about group or coset
structure of $\Vstar$. In case $V$ is irreducible, we have that
$\overline \Vstar=V$, e.g., when $V$ is toric or shifted toric. In order to
solve this problem one needs to saturate $I$ with the multiplication of the
variables and then consider the radical of this saturation. Saturation removes
the points that are in $V$ but not in $\overline\Vstar$. The following
proposition states this precisely and is the essence of this section for
computations over complex numbers.

\begin{proposition}\label{prop:group-coset-iff}
  Let $I$ be an ideal in $\QQ[x_1,\dots,x_n]$, let $V:=V(I) \subseteq \CCn$ be the variety of
  $I$, and let $G \subseteq \QQ[x_1,\dots,x_n]$ be a reduced \gb of the radical of
  $I:\ideal{x_1\dots x_n}^\infty$. Then $\Vstar=\emptyset$ if and only if
  $G=\{1\}$. If $\Vstar\neq\emptyset$, then the following hold:
  \begin{enumerate}[(i)]
  \item $\Vstar$ is a subgroup of $\Cstarn$ if and only if all
    elements of $G$ are of the form $X^\alpha-X^\beta$.
  \item $\Vstar$ is a coset of a subgroup of $\Cstarn$ if and only if
    all elements of $G$ are of the form
    $c_\alpha X^\alpha-c_\beta X^\beta$ where $c_\alpha\ne 0$ and
    $c_\beta \ne 0$.
  \end{enumerate}
\end{proposition}
\begin{proof}
  In order to prove the proposition we use \cite[Chapter 4, Theorem 10
  (iii)]{cox2015ideals}, which states that over an algebraically closed field,
  we have that $\overline{V(J_1)\setminus V(J_2)} = V(J_1:J_2^\infty)$ for
  ideals $J_1$ and $J_2$. Set $J_1=I$ and $J_2=\ideal{x_1\dots x_n}$. Since
  $\Vstar= V(I) \setminus V(x_1\dots x_n)$, and $V$ is a variety over $\CC$
  which is algebraically closed, we have that
  $ \overline \Vstar = V(I:\ideal{x_1\dots x_n}^\infty)$. Then
  \begin{equation}\label{eq:ideal-saturation-var}
    I(\overline \Vstar) = I(V(I:\ideal{x_1\dots x_n}^\infty)) =
    \sqrt{I:\ideal{x_1\dots x_n}^\infty}.
  \end{equation}
  $\Vstar = \emptyset$ if and only if
  $V(I) \setminus V(x_1\dots x_n) = \emptyset$, which is the case if
  and only if $V(I) \subseteq V(x_1\dots x_n)$. This happens if and
  only if $I(V(I)) = I(V(x_1\dots x_n))$, if and only if
  $\ideal{x_1,\dots,x_n} \subseteq \sqrt{I}$, which is the case if and
  only if some product of the variables is in $I$, or equivalently
  $I:\ideal{x_1\dots x_n}^\infty = \ideal{1}$, i.e., $G={1}$.

  For proving (i), let $\Vstar$ be a group. By Proposition
  \ref{prop:groupstructure}, we have that all the elements of every
  \gb of $I(\overline \Vstar)$ are of the form $X^\alpha -
  X^\beta$. But according to (\ref{eq:ideal-saturation-var})
  and the assumption that $G$ is a \gb of
  $\sqrt{I:\ideal{x_1\dots x_n}^\infty}$, this condition holds.

  For the converse, let the elements of $G$ have the desired form. Then
  $V(\sqrt{I:\ideal{x_1\dots x_n}^\infty})^*$ is obviously a group. But
  $V(\sqrt{I:\ideal{x_1\dots x_n}^\infty})^*=V(I:\ideal{x_1\dots x_n}^\infty)^*$
  and therefore the latter is a group. Now using \cite[Chapter 4, Theorem 10
  (iii)]{cox2015ideals}, we have that
  $V(I:\ideal{x_1\dots x_n}^\infty)^* = (\overline{V(I)\setminus V(x_1\dots
    x_n)})^*$. One can easily check that the latter is equal to $V(I)^*$. Hence
  $\Vstar=V(I)^*$ is a group and we are done.

  The proof of part (ii) is analogous to that of part (i) above.
\end{proof}

For the rest of this section, we present the monomial parametrization of a torus
and state propositions that allow one to find the cosets of a torus as
irreducible components of a binomial variety using roots of unity. Readers
primarily interested in our algorithms in Section~\ref{sec:computations} can
safely skip this part.

We start with introducing the monomial parametrization of shifted toric
varieties. Let $T\subseteq \Kstarn$ be a torus of dimension $m$, hence
$T \simeq \Kstarm$, and let $x_0\in \Kstarn$. Following the monomial
parametrization of a torus given in \cite[Corollary
2.6]{eisenbud-sturmfels-binomials}, one can see that the coset $x_0T$ of $T$ can
be seen as the image of the following monomial map, which is the monomial
parametrization of $x_0T$.
\begin{displaymath}
  \textstyle
  \varphi_{(x_0,A)}: \Kstarm\to \Kstarn,\quad
  \varphi_{(x_0,A)}(t_1,\ldots,t_m)=\left((x_0)_1\prod\limits_{i=1}^mt_i^{A_{i1}},\ldots,
    (x_0)_n\prod\limits_{i=1}^mt_i^{A_{in}}\right),
\end{displaymath}
where $A\in\ZZ^{d\times n}$is a rank $m$ matrix. Note that while the matrix $A$
is not unique, it only depends on $T$ and not on $x_0$. In particular, $T$ is
its own coset with respect to the unity $\bold 1 := (1,\ldots,1) \in \Kstarn$.
Note that if $B\in\ZZ^{d\times n}$ is another matrix such that $T$ equals the
image of $\varphi_{(\bold 1,B)}$, then $B$ corresponds to a re-parametrization of
$T$.

\begin{example}
  $V(xy-1)\cap\left(\Cstar\right)^2$ can be seen as the image of
  $\varphi_{(\bold 1,(1,-1))} :\Cstar \to (\Cstar)^2$ with
  $\varphi_{(\bold 1,(1,-1))} (t)= (t,t^{-1})$ or as the image of
  $\varphi_{(\bold 1,(-1,1))}:\Cstar \to (\Cstar)^2 $ with
  $\varphi_{(\bold 1,(-1,1))} (t)= (t^{-1},t)$.
\end{example}

\begin{proposition}\label{lemma:Roots1}
  If $G\subseteq \Cstarn$ is a group and reducible into $r \in \NN$
  cosets of a torus $T$, then there exist
  $y_1,\ldots, y_r \in \Cstarn$ whose coordinates are roots of unity
  and $G = \bigcup_{i=1}^r y_iT$ is the irreducible
  decomposition of $G$.
\end{proposition}
\begin{proof}
  If $G=T$, then it is its own coset with respect to $\bold 1$ and we are done.
  Otherwise, let $S_i = \tilde y_iT$ be a proper coset of $T$ and suppose that
  there is no $\xi\in\Cstarn$ such that the coordinates of $\xi$ are roots of
  unity and $S_i=\xi T$. This means that for all such $\xi$ the image of
  $\varphi_{(\tilde y_i,T)}$ is different from the image of
  $\varphi_{(\xi, T)}$. Hence there exists $t\in\Cstarm$ such that for all
  $s\in\Cstarm$ one has $\varphi_{(\tilde y_i,T)}(t)\ne \varphi_{(\xi,T)}(s)$.
  In other words, there exists $t\in\Cstarm$ such that for all $s \in \Cstarm$
  one has $\tilde y_i t^{A} \ne \xi s^{A}$ for some $i\in \{1,\ldots,n\}$. As
  the coordinates of $\xi$ are roots of unity, there is a natural number $N$
  such that $\xi^N = \bold 1$. Therefore, for all $s\in\Cstarm$ one has
  $\tilde y_i^{N} t^{NA} \ne s^{NA}$ for some $i\in \{1,\ldots,n\}$. As the
  image of $\varphi$ is invariant under $A \mapsto NA$, the cosets
  $\tilde y_i^NT$ and $T$ are distinct. By using a similar argument and
  induction, one can prove that $T, \tilde y_i^NT$, $\tilde y_i^{2N}T, \ldots$
  are distinct. As $G$ is closed under multiplication, it contains all these
  cosets. However, this contradicts the assumption that $G$ is reducible into a
  torus and a finite number of its cosets.
\end{proof}

\begin{remark}
  Let $S_1$ and $S_2$ be two cosets of a torus $T
  \subseteq\Cstarn$. The coset $S_1$ is called the complex conjugate
  of the coset $S_2$, written $S_1=S_2^+$, when every point of $S_1$
  is the complex conjugate of a point of $S_2$ and every point of
  $S_2$ is the complex conjugate of a point of $S_1$. As the complex
  conjugate is an automorphism of $\Cstarm$, $S_1$ is the complex
  conjugate of $S_2$ if and only if $S_1$ contains the complex
  conjugate of some point of $S_2$. A pair $S_1,S_2$ is called a pair
  of complex conjugates if $S_1=S_2^+$.  If $G\subseteq \Cstarn$ is
  reducible into a finite number of cosets of a torus then they come
  in pairs of complex conjugates.  To see this, denote the toric
  component of $G$ by $T$. If $G=T$, then clearly $G=G^+$. Suppose
  that $G$ contains a proper coset $S$ of $T$. By
  Lemma~\ref{lemma:Roots1} there is a point $\xi\in S$ whose
  coordinates are roots of unity. Then $\xi\xi^+=\bold 1$. As $G$ is a
  group, $\xi^+$ is an element of $G$. As $S=\xi T$ and
  $\xi^+T=(\xi T)^+=S^+$, we conclude that $S^+$ is contained in $G$.
\end{remark}

\begin{proposition}\label{lemma:Roots2}
  Let $P=(X/\xi)^u- (X/\xi)^v\in\CC[X]$ be a non-zero irreducible polynomial,
  where the coordinates of $\xi\in\Cstarn$ are roots of unity. Then for all $i$
  in $\NN$ there exist $g_i$ in $\CC[X]$ and $\alpha_i$, $\beta_i$ in $\NN^n$
  such that $g_iP=X^{i\alpha_i}-X^{i\beta_i}$.
\end{proposition}
\begin{proof}
  Note that $P=(X^{u}-\xi^{u-v}X^v)/\xi^u$.  As the coordinates of
  $\xi$ are roots of unity, $\gamma=\xi^{u-v}$ is also a root of
  unity. Let $m$ be the smallest positive integer such that
  $\gamma^m=1$. As
  $\mathfrak U_m =\{\gamma,\gamma^2,\ldots,\gamma^m\}$ is a group of
  roots of unity of order $m$, it is clear that
\[
(X^{u}-\gamma X^{v})(X^{u}-\gamma^2 X^{v})\cdots(X^{u}-\gamma^m
X^{v})=(X^{mu}-X^{mv}).
\]
Hence one can take
\[
  g_1=\xi^u(X^{u}-\gamma^2 X^{v})(X^{u}-\gamma^3
  X^{v})\dots (X^{u}-\gamma^m X^{v}).
\]
Substituting $\mathfrak U_m$ with the group $\mathfrak U_{mi}$, i.e.,
the group of roots of unity of order $mi$, and following the steps for
constructing $g_1$ accordingly, one can construct $g_i$ for all $i \geq 2$.
\end{proof}

\section{Algorithmic Classification of Biomodels}\label{sec:computations}
We want to apply our concept of shifted toricity to biomodels focusing on the
BioModels\footnote{\url{https://www.ebi.ac.uk/biomodels/}} repository of
mathematical models of biological and biomedical systems \cite{BioModels2015a}.
The BioModels repository uses the \emph{Systems Biology Markup Language (SBML)}
\cite{10.1093/bioinformatics/btg015,10.1042/bst0311472}. SBML is a
representation format, based on XML, for communicating and storing computational
models of biological processes. It is a free and open standard with widespread
software support and a community of users and developers. SBML models have been
typically created in the context of numerical computations or simulations and
must be processed carefully with symbolic computation. For instance, numerical
values, like reaction rate constants, contained in the models are often
represented as truncated fixed point floats, and the available SBML parsers
possibly introduce further rounding errors when implicitly performing
substitutions with those values. Such issues are addressed by
ODEbase\footnote{\url{http://odebase.cs.uni-bonn.de/}}, which provides
pre-processed versions of BioModels for use in symbolic computation. We consider
here \emph{all} models from ODEbase where the vector field of the ODE is
polynomial over $\QQ$ after application of certain SBML-specific rules and
substitution of parameter values. This amounts to a total of 129 models
considered in this article.

Following a convention often used in publications on chemical reaction network
theory in the context of symbolic computation, ODEbase replaces names of species
concentrations by more abstract names $x_i$ using numbers as indices. With the
application of SBML rules some of those $x_i$ vanish in the ODEbase toolchain.
We therefore consider, more abstractly, \emph{ordered sets $X$ of variables},
tacitly assuming that the order establishes a mapping between indeterminates in
$\QQ[X]$ and coordinates in $K^{|X|}$. As a matter of fact, the variables will
also vanish during our own algorithms discussed throughout this section. The
following example illustrates this.

\begin{example}[\texttt{BIOMD0000000198}]\label{ex:b198}
  Consider the following system in $\QQ[x_{2}, \dots, x_{10}]$:
  \begin{multline*}
    F=\{
    -350 x_{2}+800 x_{3}, 350 x_{2}-1650 x_{3}, 4250 x_{3}-100 x_{4}+x_{5}, 100
    x_{4}-x_{5},\\
    -350 x_{6}+800 x_{7}, 350 x_{6}-1650 x_{7}, 1700 x_{7}-5
    x_{8}+50 x_{9}, x_{10}+125 x_{8}-1330 x_{9}, -x_{10}+80 x_{9} \}.
  \end{multline*}
  From its Gröbner basis
  $G=\{x_{2}, x_{6}, x_{5}-100 x_{4}, 8 x_{8}-x_{10}, x_{7}, x_{3}, -x_{10}+80
  x_{9}\}$ we can read off that for every point in $V_\CC(F)\subseteq\CC^9$,
  e.g., the $x_2$-coordinate must be $0$. It follows that
  $V_\CC(F)^*=\emptyset$. Geometrically, $V_\CC(F)$ lives in $\CC^5$ with
  coordinates $x_4$, $x_5$, $x_8$, $x_9$, $x_{10}$. Thinking about toricity as a
  geometric notion, it makes sense to study the variety as an object in that
  lower dimensional space. Hence, consider
  \begin{displaymath}
    \hat G=G\setminus\{x_{2},
    x_{3}, x_{6}, x_{7}\}=\{ x_{5}-100 x_{4}, 8 x_{8}-x_{10},
    -x_{10}+80 x_{9}\}\subseteq\QQ[x_4,x_5,x_8,x_9,x_{10}]. 
  \end{displaymath}
  It turns out that $V_\CC(\hat G)$ is shifted toric in $\CC^5$.
\end{example}

\begin{definition}[Compatible and canonical projection spaces]\label{def:umax}
  Let $K \in \{\RR,\CC\}$, let $X$ be an ordered set of variables, and let
  $V_K \neq \emptyset$ be a variety in $K^{|X|}$. We say that a subset
  $\hat X \subseteq X$ \emph{describes a compatible projection space
    $K^{|\hat X|} \subseteq K^{|X|}$ with respect to $V_K$} if for the projection
  $\pi_{X \setminus \hat{X}}: K^{|X|} \to K^{|X \setminus \hat X|}$ into the complement of
  $K^{|\hat X|}$ we have $\pi_{X \setminus \hat{X}}(V_k) = \{\mathbf{0}\}$. In other
  words, for all $x \in X \setminus \hat{X}$ and $a \in V_K$ the $x$-coordinate of
  $a$ equals $0$. It is easy to see that there is a unique such $\hat X$ with
  minimum cardinality, which, we say, \emph{describes the canonical projection
    space with respect to $V_K$}.
\end{definition}

If $\hat X$ describes a canonical projection space and $\hat X \neq \emptyset$, then
$\pi_{\hat X}(V_K)^* \neq \emptyset$. On the other hand, if $V_K^*$ is empty, and
$\hat X$ describes a compatible but not canonical projection space, then still
$\pi_{\hat X}(V_k)^* = \emptyset$. In that latter case the intuition is that the
projection does not remove information from $V_K$ that is relevant for obtaining
$\pi_{\hat Y}(V_K)^*$ in the canonical projection space described by
$Y \subseteq X$. When $V_K$ is given as $V_K(F)$ for $F \subseteq \QQ[X]$, then we write
$\hat F = F \cap \QQ[\hat X]$, as we did with the Gröbner basis $G$ in
Example~\ref{ex:b198}.

The principle domain of interest for us is $\RR$-space, where, e.g.,
concentrations of species are located in the interior of the first orthant. In
the literature there has been considerable attention to $\CC$-space. We
therefore start our algorithmic considerations over $\CC$ in
Subsection~\ref{se:comp-c}, and then turn to $\RR$ in
Subsection~\ref{se:comp-r}.

In Example~\ref{ex:b198} we could conclude that $V_\CC(\hat F)^*$ is shifted
toric because $\hat F$ consists of binomials of the characteristic shape
according to Proposition~\ref{prop:group-coset-iff}, and one can easily see from
its linearity that it generates a prime ideal over $\CC$. As prime ideal
decomposition is related to polynomial factorization, decomposition or even
primeness tests over our fields $\CC$ and $\RR$ of interest are not well
supported in off-the-shelf computer algebra systems. In our algorithms we
therefore limit ourselves to the properties ``group'' and ``coset'' rather than
``toric'' and ``shifted toric''. Nevertheless, we will consider prime
decompositions over $\QQ$, which are well supported in software and provide at
least partial decompositions over $\CC$ and $\RR$. Note that for us the relevant
notion is prime decomposition in contrast to primary decomposition, as the
former corresponds to the irreducibility of the corresponding varieties.

\begin{example}[Comparison of $V_\CC^*$ with $V_\RR^*$]
  For $F_1=\langle x^2+2\rangle$, $F_2=\langle (x^2-1)(x^2+2)\rangle$ the
  following holds:
  \begin{enumerate}[(i)]
  \item $V_\CC(F_1)^*=\{i\sqrt{2},-i\sqrt{2}\}$ is a coset in $\CC^*$, because
    $(-i\sqrt{2})^{-1}V_\CC(F_1)^*=\{1, -1\}$ is a group. In contrast,
    $V_\RR(F_1)^*=\emptyset$ is not a coset in $\RR^*$.
  \item $V_\CC(F_2)^*=\{1,-1,i\sqrt{2},-i\sqrt{2}\}$ is a coset in $\CC^*$ if
    and only if it is group, due to $1\in V_\CC(F_2)^*$. This is not the case
    because it is not closed under multiplication:
    $(i\sqrt{2})(-i\sqrt{2})=2\notin V_\CC(F_2)^*$. In contrast,
    $V_\RR(F_2)^*=\{1,-1\}$ is a group.
  \end{enumerate}
\end{example}

\subsection{Classification over $\CC$}\label{se:comp-c}
Our methods over $\CC$ are, naturally, based on Gröbner bases
\cite{Buchberger:65a,Faugere:99a,Faugere:02a}, for which we rely on the
commercial computer algebra system Maple. We generally leave it to Maple to find
a good term order. Our classifications hold over any algebraically closed
extension field of the coefficient field $\QQ$, including $\CC$ as well as, e.g.,
the countable algebraic closure of $\QQ$.

With the following discussion of Algorithm~\ref{alg:pacc} we will introduce some
textbook facts from commutative algebra as lemmas, together with short proofs.
The algorithm recognizes for a given ideal basis $F$ whether $V_\CC(F)^*$ is a
coset.

\begin{algorithm}
  \caption{$\ProjectAndClassify_\CC$\label{alg:pacc}}
  \begin{algorithmic}[1]
    \REQUIRE 1.~~$X$, a finite ordered set of variables;\quad
    2.~~$F\subseteq\QQ[X]$ finite and non-empty\smallskip

    \ENSURE
    1.~$\hat{X}\subseteq X$;\quad
    2.~$\gamma\in\{\tGG, \tCC, \tOO, \tXX, \tgg, \tcc, \too,
    \txx\}$
    \smallskip
    
    $\hat{X}$ describes a compatible projection space with respect to $V_\CC(F)$.
    The letter $\gamma$ classifies $V_\CC(F)^*$ in $\hat X$-space, using upper
    case when $\hat X=X$:
    \smallskip
    
    \begin{center}
      $\tGG$/$\tgg$ -- $V_\CC(F)^*$ is a group;\quad
      $\tCC$/$\tcc$ -- $V_\CC(F)^*$ is a proper coset;\quad
      $\tOO$/$\too$ -- $V_\CC(F)^*=\emptyset$;\quad
      $\tXX$/$\txx$ else.
    \end{center}
    \smallskip

    \STATE{$G:=\GB(F)$}
    \STATE{$X':=G\cap X$}
    \STATE{$\hat X:=X\setminus X'$}
    \STATE{$\hat G:=G\setminus X'$}
    \COMMENT{$\langle\emptyset\rangle=\langle0\rangle$}
    \STATE $\tilde G:=\operatorname{Radical}(\operatorname{Saturate}(
    \hat G, \prod\hat X))$
    \COMMENT{$\tilde G$ is a Gröbner basis}
    \STATE{$\gamma:=\Classify_\CC(\hat X, \tilde G)$}
    \IF{$\hat X\neq X$}
    \STATE{convert $\gamma$ to a lower case letter}
    \ENDIF
    \RETURN{$\hat X$, $\gamma$}
  \end{algorithmic}
\end{algorithm}

In line 1 we compute a Gröbner basis $G$ of $F$. Recall from the previous
section that we generally consider reduced Gröbner bases. We may safely assume
that $\vars(G)\subseteq\vars(F)$. In line 2, the variables $X'$ occurring as elements of
$G$ are exactly those that must be zero for all points in $V_\CC(F)$. Removing
$X'$ from $X$ in line 3 yields $\hat X$ which describes a compatible but not
necessarily canonical projection space according to Definition~\ref{def:umax},
as the following example illustrates.

\begin{example}\label{ex:heu}
  Consider $F = \{x_2^2, x_1+x_2, x_2+x_3+1\}$. Then for $\hat X \subseteq X$ to
  describe the canonical projection space with respect to $V(F)$ it must not
  contain $x_2$. However, $x_2$ does not show up in the Gröbner basis
  $G = \{x_3^2+2x_3+1, x_2+x_3+1, x_1-x_3-1\}$ of $F$.
\end{example}

This idea of line 3 is to have a good heuristic method at no extra computational
cost. Removing $X'$ from $G$ in line 4 is equivalent to plugging $0$ into all
$X'$ in $G$, which in turn realizes the projection of $V_\CC(F)=V_\CC(G)$ into
$\hat X$-space. Note that we follow the convention that the empty set is a
generator of the trivial ideal \cite[Definition~1.36]{BeckerWeispfenning:93a}.
In line 5 we obtain $\tilde G$ by saturating $\hat G$ and subsequently taking
the radical. In line 6 we call Algorithm~\ref{alg:cc} in order to apply
Proposition~\ref{prop:group-coset-iff} with $I=\langle\hat G\rangle$ and $G=\tilde G$.

\begin{algorithm}
  \caption{$\Classify_\CC$\label{alg:cc}}
  \begin{algorithmic}[1]
    \REQUIRE
    1.~~$\hat X$, a finite ordered set of variables;\quad
    2.~~ $\tilde G\subseteq\QQ[\hat X]$, a Gröbner basis of a saturated radical
    ideal\smallskip
    
    \ENSURE $\gamma\in\{\tGG, \tCC, \tOO, \tXX\}$
    \smallskip
    
    \IF{$\hat X=\emptyset$ \OR $\tilde G=\{1\}$}
    \RETURN{$\tOO$}
    \ELSIF{all elements of $\tilde G$ are of the form
      $\hat X^{\alpha}-\hat X^{\beta}$ with
      $\alpha$, $\beta\in\NN^{m}$}
    \RETURN{$\tGG$}
    \ELSIF{all elements of $\tilde G$ are of the form
      $c_\alpha \hat X^{\alpha}-c_\beta \hat X^{\beta}$ with
      $\alpha$, $\beta\in\NN^{m}$, $c_\alpha\neq0$, $c_\beta\neq0$}
    \RETURN{$\tCC$}
    \ELSE
    \RETURN{$\tXX$}
    \ENDIF
  \end{algorithmic}
\end{algorithm}

In line 1 of Algorithm~\ref{alg:cc}, if $\hat X=\emptyset$, then we are in
zero-dimensional $\CC$-space and certainly $V_\CC(\hat G)^*=\emptyset$.
Otherwise $\tilde G=\{1\}$ is an equivalent criterion for $V_\CC(\hat G)^*=\emptyset$
by Proposition~\ref{prop:group-coset-iff}. From line 3 on we know that
$V_\CC(\hat G)^*\neq\emptyset$ and apply in line 3 and line 5 the criteria from
part (i) and (ii) of Proposition~\ref{prop:group-coset-iff}, respectively. In
the negative case we return $\tXX$ in line 8.

This takes us back to line 7 of Algorithm~\ref{alg:pacc}. For convenience, we
patch the classification letter $\gamma$ from upper case to lower case when proper
projection has taken place. That information could alternatively be
reconstructed by comparing $X$ with $\hat X$, which is returned in line 10 along
with $\gamma$.

\begin{example}[\texttt{BIOMD0000000519}]
  Consider $F=\{f_1,f_2,f_3\}\subseteq\QQ[X]$, where  $X=\{x_1,x_2,x_3\}$:
  \begin{align*}
    f_1 ={}& - 110569195060524661790966049 x_{1}^{2} - 110569195060524661790966049
           x_{1} x_{2}\\
         & {}- 110569195060524661790966049 x_{1} x_{3} +
           8268303407262959414915925880 x_{1},\\
    f_2 ={}&  - 39340519602534770292542037060 x_{1}^{2} -
           64716470904160708181625699581 x_{1} x_{2}\\
         &  {}- 39340519602534770292542037060 x_{1} x_{3} +
           4720862352304172435105044447200 x_{1}\\
         &  {}- 25375951301625937889083662521 x_{2}^{2} -
           25375951301625937889083662521 x_{2} x_{3}\\
         &  {}+ 1783712878395505546690039502520 x_{2},\\
    f_3 ={}& - 40542202233642354036972112493 x_{1} x_{2} -
           40542202233642354036972112493 x_{2}^{2}\\
         &  {}- 40542202233642354036972112493 x_{2} x_{3} +
           4865064268037082484436653499160 x_{2}\\
         & {}- 1101385347722460000000000000000 x_{3}.
  \end{align*}
  We obtain $\hat X=X$ and $\hat G=G=\{\hat g_1,\dots,\hat g_4\}$. For space
  reasons, we present $\hat g_1$, \dots,~$\hat g_4$ with approximate coefficients
  here:
  \begin{align*}
    \hat g_1&\approx 5.72\times 10^{41} x_{3}
            -1.05\times 10^{42} x_{2}+
            1.47\times 10^{42} x_{1},\\
    \hat g_2&\approx 3.63\times 10^{80}x_1^2-6.37\times10^{80}x_1,\\
    \hat g_3&\approx 8.89\times 10^{67}x_1x_2-2.44\times10^{69}x_1,\\
    \hat g_4&\approx 2.34\times 10^{111}x_2^2-9.39\times 10^{112}x_1-5.82\times 10^{112}x_2.
  \end{align*}
  Notice that $g_4$ is not binomial. After saturation we obtain
  $\tilde G=\{\tilde g_1,\tilde g_2,\tilde g_3\}$ with
  \begin{displaymath}
    \tilde  g_1 = \hat g_2,\quad
    \tilde g_2 = \hat g_3,\quad
    \tilde g_3 \approx 2.66\times 10^{92}x_3-1.21\times 10^{94},
  \end{displaymath}
  which is classified as $\gamma=\tCC$. Again, for $\tilde g_3$ we computed exactly
  but present here only approximate coefficients.
\end{example}

It is important to understand that, although we are using tools from ideal
theory, our results in Section~\ref{sec:preliminaries} clarify that our
classification solely depends on geometry. In particular, results are invariant
with respect to the input ideal basis $F$ in Algorithm~\ref{alg:pacc}.

In Appendix~\ref{app:cb} we discuss practical aspects of our implementation and
give in Table~\ref{tab:cbc} classification results from applying
Algorithm~\ref{alg:pacc} to the 129 models introduced at the beginning of this
section. We also address there the quality of our heuristic method for computing
the description $\hat X$ of the projection space in Algorithm~\ref{alg:pacc}.
For our discussion here we note that our algorithm terminates within a time
limit of 6 hours per model on 104 out of the 129 models. We obtain 2 $\tGG$, 20
$\tCC$, and 6 $\tcc$, which can be summarized as $V_\CC(\tilde G)^*$ forming a
coset. Furthermore we have 4 $\tOO$ and 42 $\too$, i.e.,
$V_\CC(\tilde G)^*=\emptyset$. The rest is 29 $\tXX$ and one single $\txx$. In terms of
percentages of the 104 successful computations this yields the following
picture: 
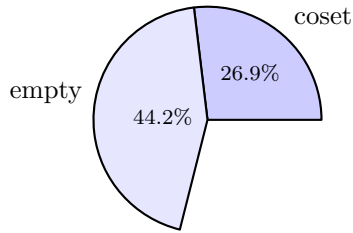
\begin{center}
  \begin{tikzpicture}
    \pie[color={blue!20, blue!10, red}, radius=1.5, sum=100, before
    number=\footnotesize, after number=\%] {
      26.9/coset,
      44.2/empty
    }
  \end{tikzpicture}
\end{center}

We are now going to turn to prime decompositions over $\QQ$ of the generating
ideals $F$ of our varieties $V_\CC(F)$. Recall that shifted toricity requires,
in addition to the coset structure of $V_\CC(F)^*$, irreducibility of
$V_\CC(F)$, which in turn corresponds to prime decompositions of $F$ even over
$\CC$. From that point of view, our decompositions considered here are only a
heuristic step into the right direction. On the other hand, the following
example suggests that beyond the irreducibility issue, prime decompositions over
$\QQ$ can improve our hit rate on cosets.

\begin{example}[\texttt{BIOMD0000000359}]\label{ex:b359}
  Consider $F\subseteq\QQ[X]$, where $X=\{x_1,\dots,x_7, x_9\}$:
  \begin{multline*}
    F=\{-125 x_{1} x_{2} - 125 x_{1} x_{5} - 11 x_{1} x_{7} + 19250 x_{3} +
    19250 x_{4},
    -5 x_{1} x_{2} + 20 x_{3} x_{7} + 770 x_{3},\\
    5 x_{1} x_{2} - 20 x_{3} x_{7} - 1190 x_{3},
    250 x_{1} x_{5} - 300 x_{4} x_{6} + 21000 x_{3} - 38500 x_{4} + x_{9},\\
    -2500 x_{1} x_{5} - 27 x_{5} x_{6} + 385000 x_{4} + 10 x_{7},
    -3000 x_{4} x_{6} - 27 x_{5} x_{6} + 10 x_{7} + 10 x_{9},\\
    -220 x_{1} x_{7} - 10000 x_{3} x_{7} + 27 x_{5} x_{6} - 10 x_{7},
    11 x_{1} x_{7}, 1000 x_{3} x_{7} + 300 x_{4} x_{6} - x_{9}\}.
  \end{multline*}
  Applying Algorithm~\ref{alg:pacc} to $X$ and $F$ yields
  $\hat X=\{x_{1}, x_{2}, x_{4},\dots, x_{7}\}$,
  $\hat G=\{x_{1} x_{7}, x_{4} x_{7}, x_{1} x_{5}-154 x_{4}, x_{1} x_{2}, x_{4}
  x_{6}, 27 x_{5} x_{6}-10 x_{7}, x_{2} x_{4}\}$, and the saturated radical
  basis $\tilde G=\{1\}$. The classification result is $\hat X$ together with
  $\gamma=\too$.

  The following is a prime decomposition of $F$ over $\QQ$:
  \begin{multline*}
    \mathcal{P}=
    (
    \{x_{1}, x_{3}, x_{4}, x_{9}, -27 x_{5} x_{6}+10 x_{7}\},
    \{x_{1}, x_{3}, x_{4}, x_{5}, x_{7}, x_{9}\},
    \{x_{1}, x_{3}, x_{4}, x_{6}, x_{7}, x_{9}\},\\
    \{x_{2}, x_{3}, x_{4}, x_{5}, x_{7}, x_{9}\},
    \{x_{2}, x_{3}, x_{6}, x_{7}, x_{9}, x_{1} x_{5}-154 x_{4}\}
    ).
  \end{multline*}
  Considering each prime component individually yields respective compatible
  subsets of variables
  $\hat{\mathcal{X}}=(\{x_{5}, x_{6}, x_{7}\}, \emptyset, \emptyset, \emptyset,
  \{x_{1}, x_{4}, x_{5}\})$ and Gröbner bases
  \begin{displaymath}
    \mathcal{\tilde G}=\hat{\mathcal{G}}=
    (\{-27 x_{5} x_{6}+10 x_{7}\}, \{0\}, \{0\}, \{0\}, \{x_{1} x_{5}-154
    x_{4}\}),
  \end{displaymath}
  which are already saturated. Application of Algorithm~\ref{alg:cc} to pairs of
  elements of $\hat{\mathcal{X}}$ and $\mathcal{\tilde G}$ yields
  $\Gamma=(\tcc,\too,\too,\too,\tcc)$. This tells us that $V_\CC(F)^*$ has two
  components, which live in different 3-dimensional subspaces of $\CC^8$. Both
  of them are cosets.
\end{example}

\begin{algorithm}[h]
  \caption{$\DecomposeAndClassify_\CC$\label{alg:dacc}}
  \begin{algorithmic}[1]
    \REQUIRE 1.~~$X$, a finite ordered set of variables;\quad
    2.~~$F\subseteq\QQ[X]$ finite and non-empty\smallskip
    
    \ENSURE 1.~~$\mathcal{P}\in\wp(\QQ[X])^k$;\quad
    2.~~$\hat{\mathcal{X}}\in\wp(X)^k$;\quad
    3.~~$\Gamma\in\{\tGG, \tCC, \tOO, \tXX, \tgg,
    \tcc, \too, \txx\}^k$ \smallskip

    \begin{sloppypar}
      $\mathcal{P}=(P_1,\dots,P_k)$ are Gröbner bases of a prime decomposition over
      $\QQ$ of $\langle F\rangle$. In
      ${\hat{\mathcal{X}}=(\hat X_1,\dots,\hat X_k)}$, $\hat X_i$ describes a
      compatible projection space with respect to $V_\CC(P_i)$. In
      $\Gamma=(\gamma_1,\dots,\gamma_k)$, the letter $\gamma_i$ classifies
      $V_\CC(P_i)^*$ in $\hat X_i$-space, using upper case when $\hat X_i=X$:
    \end{sloppypar}
    \smallskip

    \begin{center}
      $\tGG$/$\tgg$ -- $V_\CC(P_i)^*$ is a group;\quad
      $\tCC$/$\tcc$ -- $V_\CC(P_i)^*$ is a proper coset;\quad
      $\tOO$/$\too$ -- $V_\CC(P_i)^*=\emptyset$;\quad
      $\tXX$/$\txx$ else.
    \end{center}
    \smallskip
    
    \STATE
    $\mathcal{P}=(P_1,\dots,P_k):=\operatorname{PrimeDecomposition_\QQ}(F)$
    \COMMENT{$P_1$, \dots, $P_k$ are Gröbner bases}
    \FOR{$i=1,\dots,k$}
    \STATE{$X_i':=P_i\cap X$}
    \STATE{$\hat X_i:=X_i\setminus X_i'$}
    \STATE{$\hat P_i:=P_i\setminus X_i'$}
    \COMMENT{$\langle\emptyset\rangle=\langle0\rangle$}
    \STATE $\tilde P_i:= \operatorname{Saturate}(\hat P_i, \prod\hat X_i)$
    \COMMENT{$\tilde P_i$ is a Gröbner basis; the product runs over the set $\hat
      X_i$}
    \STATE{$\gamma_i:=\Classify_\CC(\hat X_i, \tilde P_i)$}
    \COMMENT{call Algorithm~\ref{alg:cc}}
    \IF{$\hat X_i\neq X$}
    \STATE{convert $\gamma_i$ to a lower case letter}
    \ENDIF
    \ENDFOR
    \RETURN{$\mathcal{P}$, $(\hat X_1,\dots, \hat X_k)$,
      $(\gamma_1,\dots,\gamma_k)$}
  \end{algorithmic}
\end{algorithm}

Algorithm~\ref{alg:dacc} formalizes the approach outlined in
Example~\ref{ex:b359}. We use the Weierstrass $\wp$ for power sets. The algorithm
starts with the computation of a prime decomposition in line 1. We have
\begin{displaymath}
  \bigcup_{i=1}^kV_\CC(P_i) = V_\CC(F) = V_\CC(\sqrt{\langle F \rangle}) \quad
  \text{and}\quad
  \bigcap_{i=1}^k\langle P_i \rangle = \sqrt{\langle F \rangle}.
\end{displaymath}
Note that the obtained prime ideals $\langle P_i \rangle$ are also radical.

\begin{lemma}\label{lem:primeisrad}
  Let $I\subseteq\QQ[X]$ be a prime ideal. Then $I$ is a radical ideal.
\end{lemma}

\begin{proof}
  Let $f^s\in I$. We show by induction on $s$ that $f\in I$. If $s=1$, then we
  are done. Otherwise consider $f^s=ff^{s-1}\in I$. Since $I$ is prime, we have
  $f\in I$ or $f^{s-1}\in I$. In the latter case, $f\in I$ by the induction
  hypothesis.
\end{proof}

In lines 3--10, Algorithm~\ref{alg:dacc} follows in a for-loop essentially
Algorithm~\ref{alg:pacc} for each prime component basis $P_i$. In line 5 we note
that $\langle \hat P_i\rangle$ is prime by the following lemma.

\begin{lemma}
  Let $G\subseteq\QQ[X]$ be a reduced Gröbner basis of a prime ideal
  $\langle G\rangle$, and let $x\in G\cap X$. Then
  $\langle G\setminus\{x\}\rangle$ is prime.
\end{lemma}

\begin{proof}
  Notice that $\langle G\setminus\{x\}\rangle$ is the elimination ideal
  $\langle G\rangle_x=\langle G\rangle\cap\QQ[X\setminus\{x\}]$. Let
  $fg\in\langle G\rangle_x\subseteq\langle G\rangle$. Then
  w.l.o.g.~$f\in\langle G\rangle$, because $\langle G\rangle$ is prime. Since
  $x$ does not occur in $fg$, it does not occur in $f$ either. Hence
  $f\in \langle G\rangle_x$.
\end{proof}

It follows that $\langle \hat P_i\rangle$ is also radical by
Lemma~\ref{lem:primeisrad}. When computing $\tilde P_i$ in line 6 primality is again
preserved, as the following lemma shows.

\begin{lemma}
  Let $I\subseteq\QQ[X]$ be a prime ideal, and let $f\in\QQ[X]$. Then
  $I:\langle f\rangle^\infty$ is a prime ideal.
\end{lemma}

\begin{proof}
  Recall that $I\subseteq I:\langle f\rangle^\infty$. Let
  $gh\in I:\langle f\rangle^\infty$. We must show that
  $g\in I:\langle f\rangle^\infty$ or $h\in I:\langle f\rangle^\infty$. By
  definition there is $s\in\NN$ such that $f^sgh\in I$. If $f^s\in I$, then
  $I:\langle f\rangle^\infty=\langle1\rangle$, which is prime. Otherwise
  $gh\in I$ and therefore $g\in I\subseteq I:\langle f\rangle^\infty$ or
  $h\in I\subseteq I:\langle f\rangle^\infty$.
\end{proof}

We once more call Lemma~\ref{lem:primeisrad} to obtain that $\langle
\tilde P_i\rangle$ is also radical. Therefore, the radical ideal computation in line
5 of Algorithm~\ref{alg:pacc} is not necessary here.

In Appendix~\ref{app:cp} we discuss practical aspects of our computations and
give in Table~\ref{tab:pc} classification results using Algorithm~\ref{alg:dacc}
on the 129 models introduced at the beginning of this section. We succeed on 105
out of the 129 models within a time limit of 6 hours per model. This yields 3426
prime components to test altogether. We obtain 2 $\tGG$, 22 $\tCC$, and 1085
$\tcc$, which can be summarized as $V_\CC(\tilde P_i)^*$ forming a coset.
Furthermore, we have 2242 $\too$, i.e., $V_\CC(\tilde P_i)^*=\emptyset$. The rest is
only 34 $\tXX$ and 41 $\txx$. Again we visualize these results in terms of
percentages of the total of 3426 prime components:
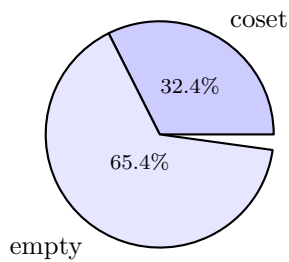
\begin{center}
  \begin{tikzpicture}
    \pie[color={blue!20, blue!10}, radius=1.5, sum=100, before
    number=\footnotesize, after number=\%] {
      32.4/coset,
      65.4/empty
    }
  \end{tikzpicture}
\end{center}
Recall that our selection from the BioModels repository presented here is
essentially complete with respect to polynomial examples. This comes with the
disadvantage that our data is somewhat dominated by \texttt{BIOMD0000000281},
which contributes 1008 $\tcc$ and 2136 $\too$. We have verified that the ideal
dimensions for the 1008 $\tcc$ components are positive, pointing at non-trivial
coset structures in contrast to isolated points. For the sake of scientific
rigor we also present the statistics without \texttt{BIOMD0000000281}:
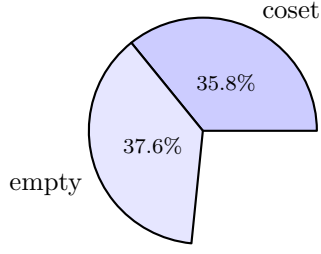
\begin{center}
  \begin{tikzpicture}
    \pie[color={blue!20, blue!10}, radius=1.5, sum=100, before
    number=\footnotesize, after number=\%] {
      35.8/coset,
      37.6/empty
    }
  \end{tikzpicture}
\end{center}

\subsection{Classification over $\RR$}\label{se:comp-r}
Our primary tool over $\RR$ is real quantifier elimination
\cite{Tarski:48a,Collins:75,CollinsHong:91,Weispfenning:88a,Grigoriev:88a,Weispfenning:97b}.
We use implementations by the fourth author and his students
\cite{DolzmannSturm:97c,Sturm:99a} in Redlog
\cite{DolzmannSturm:97a,SeidlSturm:03c,Sturm:06a,Sturm:07a,Sturm:17a}, which is
integrated with the open-source computer algebra system Reduce
\cite{Hearn:67a,Hearn:05a}. Our strategy is to apply virtual substitution
methods \cite{Weispfenning:97b,KostaSturm:16a,Kosta:16a} for quantifier
elimination within the relevant degree bounds and fall back into partial
cylindrical algebraic decomposition \cite{DolzmannSeidl:04a,Seidl:06a} with
subproblems where this is not possible. Our results hold over any real-closed
field, including $\RR$ as well as, e.g, the countable field of real algebraic
numbers.

Algorithm~\ref{alg:pacr} is the real counterpart to Algorithm~\ref{alg:pacc} in
Subsection~\ref{se:comp-c}.
\begin{algorithm}
  \caption{$\ProjectAndClassify_\RR$\label{alg:pacr}}
  \begin{algorithmic}[1]
    \REQUIRE 1.~~$X$, a finite ordered set of variables;\quad
    2.~~$F\subseteq\QQ[X]$ finite and non-empty\smallskip

    \ENSURE
    1.~$\hat{X}\subseteq X$;\quad
    2.~$\gamma\in\{\tGG, \tCC, \tOO, \tXX, \tgg, \tcc, \too,
    \txx\}$
    \smallskip

    $\hat{X}$ describes the canonical projection space with respect to $V_\RR(F)$.
    The letter $\gamma$ classifies $V_\RR(F)^*$ in $\hat X$-space, using upper
    case when $\hat X=X$:
    \smallskip
    
    \begin{center}
      $\tGG$/$\tgg$ -- $V_\RR(F)^*$ is a group;\quad
      $\tCC$/$\tcc$ -- $V_\RR(F)^*$ is a proper coset;\quad
      $\tOO$/$\too$ -- $V_\RR(F)^*=\emptyset$;\quad
      $\tXX$/$\txx$ else.
    \end{center}
    \smallskip

    \STATE{$\hat X := X$}
    \STATE{$\hat F := F$}
    \FOR{$x_i\in X$}
    \IF{$\RR\models\underline{\forall}(\bigwedge_{f\in\hat F}f=0\longrightarrow
      x_i=0)$}
    \STATE{$\hat X:=\hat X\setminus\{x_i\}$}
    \STATE{$\hat F:=\hat F[x_i/0]$}
    \ENDIF
    \ENDFOR
    \STATE{$\gamma:=\Classify_\RR(\hat X, \hat F)$}
    \IF{$\hat{X}\neq X$}
    \STATE{convert $\gamma$ to a lower case letter}
    \ENDIF
    \RETURN{$\hat{X}$, $\gamma$}
  \end{algorithmic}
\end{algorithm}
In line 4 we construct for each $x_i\in X$ the first-order $\LOR$-formula, 
where $\LOR$ denotes the language of ordered rings:
\begin{displaymath}
  \textstyle
  \psi\dotequal
  \underline{\forall}\biggl(\bigwedge\limits_{f\in\hat F}f=0\longrightarrow
    x_i=0\biggl),
\end{displaymath}
The underlined universal quantifier denotes the universal closure, which
universally quantifies all variables freely occurring within its scope. Our
formula $\psi$ straightforwardly states that for all points in $V_\RR(\hat F)$
with coordinates $x_j\in X$, which occur as variables in the polynomials
$f\in\hat F$, the specific coordinate $x_i$ is zero. The if-condition
$\RR\models\psi$ expresses that $\RR$ is a model of this formula, meaning that
the formula holds in $\RR$ or, equivalently, in the model class of real closed
fields.

A \emph{real quantifier elimination procedure} computes for any given
first-order $\LOR$-formula $\varphi$ an equivalent quantifier-free
$\LOR$-formula $\varphi'$, where the variables in $\varphi'$ are a subset of the
variables freely occurring in $\varphi$. Since $\psi$ contains no free
occurrences of variables, the corresponding $\psi'$ will be variable-free and
can be easily simplified to either $\true$ or $\false$. In the former case the
if-condition holds, in the latter case it does not.

When some $x_i$ is identified to vanish in all points of $V_\RR(\hat F)$ it is
removed from $\hat X$ in line 5. Notice that in contrast to
Algorithm~\ref{alg:pacc} the final $\hat X$ describes not only a compatible but
the canonical projection space with respect to $V_\RR(F)$. Accordingly, $x_i$ is
set to zero within $\hat F$ in line 6, where $[x_i/0]$ is a postfix operator
substituting the term $0$ for the variable $x_i$ in its argument $\hat F$. From
line 9 on, Algorithm~\ref{alg:pacr} proceeds like its complex counterpart
Algorithm~\ref{alg:pacc} but using Algorithm~\ref{alg:cr} for real
classification.
\begin{algorithm}
  \caption{$\Classify_\RR$\label{alg:cr}}
  \begin{algorithmic}[1]
    \REQUIRE
    1.~~$\hat X$, a finite ordered set of variables, w.l.o.g.~$\hat
    X=\{x_1,\dots,x_n\}$;\quad
    2.~~$\hat F\subseteq\QQ[\hat X]$ finite
    \smallskip
    
    \ENSURE $\gamma\in\{\tGG, \tCC, \tOO, \tXX\}$
    \smallskip

    \STATE define operator $\Phi(t_1,\dots,t_n) :=
    (\bigwedge_{f\in \hat F}f=0)[x_1/t_1,\dots,x_n/t_n]$
    \IF{$\hat X=\emptyset$ \OR
      $\RR\nmodels\underline{\exists}(
      \bigwedge_{i=1}^nx_i\neq0\land
      \Phi(x_1,\dots,x_n))$}
    \RETURN $\tOO$
    \ENDIF
    \STATE $\tau_{\text{inv}}:=\underline{\forall}(
    \bigwedge_{i=1}^n g_i\neq0\land
    \bigwedge_{i=1}^n x_i\neq0\land\newline
    \phantom{\tau_{\text{inv}}:=\underline{\forall}(}
    \quad\Phi(g_1,\dots,g_n)\land
    \Phi(g_1x_1,\dots,g_nx_n)
    \longrightarrow\Phi(g_1x_1^{-1},\dots,g_nx_n^{-1})$
    \IF{$\RR\nmodels\tau_{\text{inv}}$}
    \RETURN{$\tXX$}
    \ENDIF
    \STATE $\tau_{\text{mult}}:=\underline{\forall}(
    \bigwedge_{i=1}^n g_i\neq0\land
    \bigwedge_{i=1}^n x_i\neq0\land
    \bigwedge_{y=1}^n y_i\neq0\land
    \Phi(g_1,\dots,g_n)\land\newline
    \phantom{\tau_{\text{mult}}:=\underline{\forall}(}
    \quad\Phi(g_1x_1,\dots,g_nx_n)\land
    \Phi(g_1y_1,\dots,g_ny_n)\longrightarrow \Phi(g_1x_1y_1,\dots,g_nx_ny_n)$
    \IF{$\RR\nmodels\tau_{\text{mult}}$}
    \RETURN $\tXX$
    \ENDIF
    \STATE{$\tau_{\text{group}}:=\Phi(1,\dots,1)$}
    \IF{$\RR\models\tau_{\text{group}}$}
    \RETURN{$\tGG$}
    \ELSE
    \RETURN{$\tCC$}
    \ENDIF
  \end{algorithmic}
\end{algorithm}
In line 1 we define $\Phi(t_1,\dots,t_n)$ to generate a first order
$\LOR$-formula which states that $(t_1,\dots,t_n)\in V_\RR(\hat F)$, where the
$t_i$ are $\LOR$-terms. In lines 2--4 we handle the case
$V_\RR(\hat F)^*=\emptyset$. Hence in line 5 we know
$V_\RR(\hat F)^*\neq\emptyset$. We are going to use the following
characterization of cosets.

\begin{proposition}\label{prop:coset}
  Let $K^*$ be a multiplicative group. Let $C\subseteq (K^*)^n$,
  $C\neq\emptyset$. Then the following are equivalent:
  \begin{enumerate}[(i)]
  \item $C$ is a coset;
  \item there exists $g_0\in(K^*)^n$ such that $g_0^{-1}C$ is a group;
  \item there exists $g_0\in C$ such that $g_0^{-1}C$ is a group;
  \item for all $g\in C$ we have that $g^{-1}C$ is a group.
  \end{enumerate}
\end{proposition}

\begin{proof}
  The equivalence between (i) and (ii) is the definition of a coset. When (ii)
  holds, then $C=g_0G$ for a group $G$, hence $g_0\cdot1\in C$, which shows
  (iii). The implication from (iii) to (ii) is obvious, and so is the
  implication from (iv) to (iii). It remains to be shown that (iii) implies
  (iv).

  Assume that $g_0\in C$ and $G=g_0^{-1}C$ is a group; equivalently $C = g_0G$.
  Let $g\in C$. Then there is $y\in G$ such that $g=g_0y$. It follows that
  $g^{-1}C=(g_0y)^{-1}C=y^{-1}g_0^{-1}C=y^{-1}G=G$.
\end{proof}

Proposition~\ref{prop:coset}(iv) yields a first-order characterization for
$V_\RR(\hat F)^*$ to be a coset, which could be informally stated as follows:
\begin{equation}
  \label{eq:AAA}
  \forall g, x, y\in(\RR^*)^n\mathord{:}\
  g\in V_\RR(\hat F)\land gx\in V_\RR(\hat F)\land
  gy\in V_\RR(\hat F)\Rightarrow gx^{-1}\in V_\RR(\hat F) \land gxy\in
  V_\RR(\hat F).
\end{equation}
As a first-order $\LOR$-formula this yields:
\begin{multline*}
  \textstyle
  \tau\dotequal\underline{\forall}\biggl(
  \bigwedge\limits_{i=1}^n g_i\neq0\land
  \bigwedge\limits_{i=1}^n x_i\neq0\land
  \bigwedge\limits_{y=1}^n y_i\neq0\land{}
  \Phi(g_1,\dots,g_n)\land
  \Phi(g_1x_1,\dots,g_nx_n)\land{}\\
  \Phi(g_1y_1,\dots,g_ny_n)
  \longrightarrow \Phi(g_1x_1^{-1},\dots,g_nx_n^{-1})\land
  \Phi(g_1x_1y_1,\dots,g_nx_ny_n)\biggr).
\end{multline*}
In the equations originating from $\Phi(g_1x_1^{-1},\dots,g_nx_n^{-1})$
principal denominators containing variables from $x_1$, \dots,~$x_n$ can be
equivalently dropped, because the left hand side of the implication requires
those variables to be different from zero. The first-order $\LOR$-formula
$\tau$ can be equivalently transformed into
$\tau_{\text{inv}}\land\tau_{\text{mult}}$ with $\tau_{\text{inv}}$ and
$\tau_{\text{mult}}$ as in line 5 and line 9 of Algorithm~\ref{alg:cr},
respectively. Therefore it is correct to exit with $\gamma=\tXX$ in line 7 or
11 when either part does not hold. This splitting into subproblems has two 
advantages. First, separate quantifier eliminations on smaller problems are
more efficient. Second, when $\tau_{\text{inv}}$ does not hold in line 6, then $
\tau_{\text{mult}}$ need not be considered at all.

When reaching line 13, we know that $V_\RR(\hat F)^*$ is a coset and apply the
following corollary, which concludes our discussion of Algorithm~\ref{alg:cr}.
\begin{corollary}\label{cor:groupwhenone}
  Let $C$ be a coset. Then $C$ is group if and only if $1\in C$.
\end{corollary}

\begin{proof}
  If $1\in C$, then $C=1^{-1}C$ is a group by Proposition~\ref{prop:coset}(iv).
  The converse implication follows from the definition of a group.
\end{proof}

\begin{remark}
  As an alternative to (\ref{eq:AAA}), Proposition~\ref{prop:coset}(iii) yields
  the following characterization of cosets, which might appear more natural
  because it is closer to the original definition of cosets:
  \begin{equation}
    \label{eq:EAA}
    \exists g\in(\RR^*)^n~\forall x, y\in(\RR^*)^n\mathord{:}\
    gx\in V_\RR(\hat F)\land
    gy\in V_\RR(\hat F)\Rightarrow gx^{-1}\in V_\RR(\hat F) \land gxy\in
    V_\RR(\hat F).
  \end{equation}
  The first difference to observe is that in (\ref{eq:EAA}) in contrast to
  (\ref{eq:AAA}) there is quantifier alternation from $\exists$ to $\forall$.
  The number of quantifier alternations is known to be a critical parameter for
  asymptotic complexity of the real quantifier elimination problem
  \cite{Grigoriev:88a,Weispfenning:88a}. Furthermore, in the presence of the
  leading existential quantifier prohibits our splitting into two independent
  smaller problems. Experimental computations on the complete dataset considered
  here have confirmed that formulation (\ref{eq:AAA}) is clearly preferable.
\end{remark}

\begin{example}[\texttt{BIOMD0000000159}]
  Consider $F\subseteq\QQ[X]$, where $X=\{x_1,\dots,x_3\}$:
  \begin{displaymath}
    F=\{- 32 x_{1} x_{2} + 3,  - x_{2} + x_{3}, 4 x_{1} - x_{3}\}.
  \end{displaymath}
  In lines 3--8 of Algorithm~\ref{alg:pacr} we consecutively apply real
  quantifier elimination to the following formulas:
  \begin{displaymath}
    \forall x_1\forall x_2\forall x_3(
    - 32 x_{1} x_{2} + 3=0
    \land - x_{2} + x_{3}=0
    \land 4 x_{1} - x_{3}=0
    \longrightarrow
    x_i=0),\quad
    i=1, \dots,3.
  \end{displaymath}
  Neither of them holds in $\RR$ so that in line 9 we enter
  Algorithm~\ref{alg:cr} with $\hat X=X$ and $\hat F=F$.
  
  In line 2 of Algorithm~\ref{alg:cr} we test
  \begin{equation}
    \label{eq:rex1}
    \exists x_1\exists x_2\exists x_3(x_1\neq0\land x_2\neq0\land x_3\neq0
    \land - 32 x_{1} x_{2} + 3=0
    \land - x_{2} + x_{3}=0
    \land 4 x_{1} - x_{3}=0).
  \end{equation}
  Real quantifier elimination \cite{Kosta:16a} confirms that (\ref{eq:rex1})
  holds in $\RR$, and extended quantifier elimination \cite{KostaSturm:16a} even
  gives us a witness
  \begin{equation}
    \label{eq:witness}
    \textstyle
    (x_1,x_2,x_3)=\left(
      \frac{\sqrt{3}}{8 \sqrt{2}},
      \frac{\sqrt{3}}{2 \sqrt{2}},
      \frac{\sqrt{3}}{2 \sqrt{2}}\right)\in V_\RR(F)^*.
  \end{equation}
  Therefore we set up $\tau_{\text{inv}}$ in line 5 as follows:
  \begin{multline}
    \label{eq:rex2}
    \tau_{\text{inv}}\dotequal
    \forall g_1\forall g_2\forall g_3
    \forall x_1\forall x_2\forall x_3(
    g_1=0\land g_2=0\land g_3=0\land
    x_1=0\land x_2=0\land x_3=0\\
    {}\land- 32 g_{1} g_{2} + 3=0
    \land - g_{2} + g_{3}=0
    \land 4 g_{1} - g_{3}=0\\
    {}\land - 32 g_1x_{1} g_2x_{2} + 3=0
    \land - g_2x_{2} + g_3x_{3}=0
    \land 4 g_1x_{1} - g_3x_{3}=0\\
    {}\longrightarrow
    {- 32 g_{1} g_{2} + 3x_1x_2}=0
    \land - g_{2}x_3 + g_{3}x_2=0
    \land 4 g_{1}x_3 - g_{3}x_1=0).
  \end{multline}
  Notice that in the three equations in last line of (\ref{eq:rex2}) we have
  equivalently dropped denominators $x_1x_2$, $x_2x_3$, and $x_1x_3$,
  respectively. The inequalities for $x_1$, \dots,~$x_3$ in first line of
  (\ref{eq:rex2}) ensure that those denominators do not vanish. In line 10,
  quantifier elimination confirms that $\tau_{\text{inv}}$ holds in $\RR$, and
  so does $\tau_{\text{mult}}$ in line 10. Thus we reach line 13 and set up the
  following formula to test whether our coset $\hat F$ is even a group:
  \begin{displaymath}
    - 32 + 3=0
    \land - 1 + 1=0
    \land 4 - 1=0.
  \end{displaymath}
  This is obviously not the case. Algorithm~\ref{alg:cr} returns \lq$\tCC$\rq{}
  in line 17, and Algorithm~\ref{alg:pacr} finally returns $\{x_1,x_2,x_3\}$ and
  \lq$\tCC$\rq{} in line 13.
\end{example}

In Appendix~\ref{app:cb} we discuss practical aspects of our implementation and
give in Table~\ref{tab:cbc} classification results from applying
Algorithm~\ref{alg:pacr} to the 129 models introduced at the beginning of this
section. For our discussion here we note that our algorithm terminates within a
time limit of 6 hours per model on 94 out of the 129 models. We obtain 20 $\tCC$
and 6 $\tcc$, which can be summarized as $V_\RR(\tilde G)^*$ forming a coset.
Furthermore we have 4 $\tOO$ and 42 $\too$, i.e., $V_\RR(\tilde G)^*=\emptyset$. The
rest is 21 $\tXX$ and one single $\txx$. In terms of percentages of the 94
successful computations this gives the following picture:
\begin{center}
  \begin{tikzpicture}
    \pie[color={red!20, red!10}, radius=1.5, sum=100,
    before number=\footnotesize] {
      27.7/coset,
      48.9/empty
    }
  \end{tikzpicture}
\end{center}

In analogy to Algorithm~\ref{alg:dacc} in Subsection~\ref{se:comp-c},
Algorithm~\ref{alg:dacr} applies prime decompositions over $\QQ$ also in the
real case.
\begin{algorithm}
  \caption{$\DecomposeAndClassify_\RR$\label{alg:dacr}}
  \begin{algorithmic}[1]
    \REQUIRE 1.~~$X$, a finite ordered set of variables;\quad
    2.~~$F\subseteq\QQ[X]$ finite and non-empty\smallskip
    
    \ENSURE 1.~~$\mathcal{P}\in\wp(\QQ[X])^k$;\quad
    2.~~$\hat{\mathcal{X}}\in\wp(X)^k$;\quad
    3.~~$\Gamma\in\{\tGG, \tCC, \tOO, \tXX, \tgg,
    \tcc, \too, \txx\}^k$ \smallskip

    \begin{sloppypar}
      $\mathcal{P}=(P_1,\dots,P_k)$ are Gröbner bases of a prime decomposition over
      $\QQ$ of $\langle F\rangle$. In
      ${\hat{\mathcal{X}}=(\hat X_1,\dots,\hat X_k)}$, $\hat X_i$ describes the
      canonical projection space with respect to $V_\RR(P_i)$. In
      $\Gamma=(\gamma_1,\dots,\gamma_k)$, the letter $\gamma_i$ classifies
      $V_\RR(P_i)^*$ in $\hat X_i$-space, using upper case when $\hat X_i=X$:
    \end{sloppypar}
    \smallskip

    \begin{center}
      $\tGG$/$\tgg$ -- $V_\RR(P_i)^*$ is a group;\quad
      $\tCC$/$\tcc$ -- $V_\RR(P_i)^*$ is a proper coset;\quad
      $\tOO$/$\too$ -- $V_\RR(P_i)^*=\emptyset$;\quad
      $\tXX$/$\txx$ else.
    \end{center}
    \smallskip
    
    \STATE
    $\mathcal{P}=(P_1,\dots,P_k):=\operatorname{PrimeDecomposition_\QQ}(F)$
    \FOR{$i=1,\dots,k$}
    \STATE{$\text{$\hat X_i$, $\gamma_i$}:=\ProjectAndClassify_\RR(X, P_i)$}
    \ENDFOR
    \RETURN{$\mathcal{P}$, $(\hat X_1,\dots, \hat X_k)$,
      $(\gamma_1,\dots,\gamma_k)$}
  \end{algorithmic}
\end{algorithm}
It starts with the computation of a prime decomposition in line 1. We
then have
\begin{displaymath}
  \bigcup_{i=1}^kV_\RR(P_i)=V_\RR(F)\quad
  \text{and}\quad
  \RR\models\bigvee_{i=1}^k\bigwedge_{p\in P_i}p=0
  \longleftrightarrow
  \bigwedge_{f\in F}f=0.
\end{displaymath}
In lines 3--4 we apply Algorithm~\ref{alg:pacr} to each component and collect
the results.

In Appendix~\ref{app:cp} we discuss practical aspects of our computations and
give in Table~\ref{tab:pc} classification results using Algorithm~\ref{alg:dacr}
on the 129 models introduced at the beginning of this section. We succeed on 88
out of the 129 models within a time limit of 6 hours per model. This yields 3390
prime components to test altogether. We obtain 2 $\tGG$, 22 $\tCC$, and 1083
$\tcc$, which can be summarized as $V_\RR(\tilde P_i)^*$ forming a coset.
Furthermore, we have 7 $\tOO$ and 2232 $\too$, i.e., $V_\RR(\tilde P_i)^*=\emptyset$.
The rest is only 18 $\tXX$ and 26 $\txx$. In left hand side picture below, we
visualize these results in terms of percentages of the total of 3390 prime
components. In the right hand side picture, we see the corresponding statistics
without \texttt{BIOMD0000000281}:
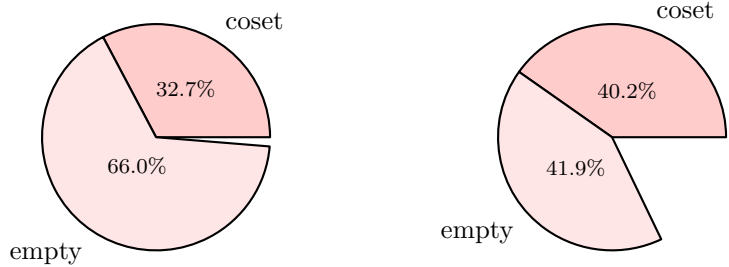
\begin{center}
  \begin{tikzpicture}
    \pie[color={red!20, red!10}, radius=1.5, sum=100, before
    number=\footnotesize] {
      32.7/coset,
      66.0/empty
    }
    \pie[color={red!20, red!10}, radius=1.5, sum=100, before
    number=\footnotesize, pos={6,0}] {
      40.2/coset,
      41.9/empty
    }
  \end{tikzpicture}
\end{center}
Recall from the discussion at the end of Subsection~\ref{se:comp-c} that we
consider the left hand side picture more adequate and add the right hand
side one for the sake of scientific rigor.
 
\section{Upper Complexity Bounds}\label{sec:complexities}
In this section we give asymptotic upper bounds on the worst case complexity of
problems addressed in this paper. In Subsection~\ref{subsec:algebraicallyclosed}
we derive bounds for recognizing toric and shifted toric varieties over
algebraically closed fields of characteristic zero. In
Subsection~\ref{subsection:realclosed} we derive corresponding bounds for toric
varieties over real closed fields. Subsection~\ref{subsection:binomiality}
finally gives bounds for the membership problem in subgroups of
$({\overline\QQ}^*)^n$, which correspond to binomial varieties.

\subsection{Toricity over Algebraically Closed Fields of Characteristic
  Zero}\label{subsec:algebraicallyclosed}
As mentioned in Section \ref{sec:preliminaries}, a torus $G$ can be represented
as the set of solutions of binomials of the form
$x_1^{a_{1i}}\cdots x_n^{a_{ni}}=1$, $1\le i\le n-m$, where $a_{ji}\in \ZZ$ and
$m$ is equal to the dimension of $G$ and every Gröbner basis of a toric variety
$\overline G$ consists only of binomials of the form
$x_1^{b_{1i}}\cdots x_n^{b_{ni}}- x_1^{c_{1i}}\cdots x_n^{c_{ni}}$ where
integers $b_{ji}$, $c_{ji}$ are non-negative \cite{sturmfels_grobner_1996}.

Let $f_1$, \dots,~$f_k\in \Zx$ and $\VV \subseteq \CCn$ be the algebraic variety
of common zeroes of $f_1$, \dots,~$f_k$. We design an algorithm which recognizes
whether $\VV$ is toric. Note that $f_1$, \dots,~$f_k$ are not necessary
binomials. To estimate the complexity of the algorithm we suppose that
$\deg(f_i)\le d$, $1\le i\le k$ and that the bit-size of each integer
coefficient of $f_1$, \dots,~$f_k$ does not exceed $L$. Invoking
\cite{chistov_algorithm_1986,grigoriev_factorization_1986} we first verify
whether $\VV$ is irreducible and $\overline{\Vstar}=\VV$ (if this is not true
then $\VV$ is not toric). The complexity of the algorithms from
\cite{chistov_algorithm_1986,grigoriev_factorization_1986} can be bounded by
$(Ld^{n^2})^{O(1)}$. Then we verify that $\Vstar$ is a group. This holds if and
only if we have the following first-order formula in the theory of algebraically
closed fields of characteristic zero:
\begin{eqnarray}\label{1}
  \textstyle
  \forall x\forall y(x\neq0\land y\neq0\land
  \bigwedge_if_i(x)=0\land \bigwedge_if_i(y)=0\longrightarrow
  \bigwedge_if_i(xy)=0\land \bigwedge_if_i(1/x)=0).
\end{eqnarray}
This can be verified via the algorithm in \cite{chistov_complexity_1984}. The
complexity of this step is bounded by $(Ld^{n^2})^{O(1)}$ as well.

So far, the algorithm has verified whether $\VV$ is toric. Now we show how to
find a system of binomial equations determining $\Vstar$. In order to do so, in
this subsection we find a set of Laurent binomials determining $\Vstar$ instead
of binomials in $\CC[X]$ whose set of solutions is $\VV$. This is because as it
has been shown by Mayr and Meyer in their seminal work, the number of binomials
and their degrees in a Gröbner basis of $I(\VV)$ can be double-exponential
\cite{mayr_complexity_1982}. Using the algorithms from
\cite{chistov_algorithm_1986,grigoriev_factorization_1986} one can produce
$m=\dim \VV$ coordinates among $x_1$, \dots,~$x_n$ which form a transcendental
basis of $\Vstar$. Without loss of generality, assume that $\{x_1,\dots,x_m\}$
be a transcendental basis. Fix an integer $j$, $m<j\le n$ and project $\Vstar$
on the space generated by $x_1$, \dots,~$x_m$, $x_j$, which is isomorphic to
$(\Cstar)^{m+1}$ invoking again \cite{chistov_complexity_1984}. Let
$W\subseteq (\Cstar)^{m+1}$ be image of the projection. Due to the choice of the
transcendental basis, we have that $\dim W=m$ and $W$ is a hypersurface.
Therefore, $W$ can be determined by a single polynomial
$h:=h_j\in \ZZ[x_1,\dots,x_m, x_j]$ (c.f. e.g., \cite{Shafarevich}).
Moreover, $\deg W\le \deg \Vstar\le d^n$; the latter follows from Bezout
inequality \cite{Shafarevich}. The algorithm from \cite{chistov_complexity_1984}
constructs $h$ and a generic point of $W$ within the complexity
$(Ld^{n^2})^{O(1)}$.

Observe that $W$ is also a group. Hence $h$ can be rewritten as a binomial of
the form $x_1^{q_{j1}}\cdots x_m^{q_{jm}}\cdot x_j^{q_j} -1 \in \Zxpm$ for suitable
relatively prime integers $q_{j1}$, \dots,~$q_{jm}$, $q_j \in \ZZ$. Doing so for every
$j$, $m<j\le n$, the algorithm yields polynomials $h_j \in \Zx$. Denote by
$H \subseteq \Cstarn$ the variety given by equations $h_j$, $m<j\le n$. Clearly,
$\Vstar\subseteq H$, $\dim H=m$ and $H$ is a group, therefore, $\Vstar$ is an
irreducible component of $H$. Moreover, $H$ is a binomial variety and hence,
$\Vstar$ is its subgroup (of a finite index). In particular,
$(1,\dots,1)\in \Vstar$, and every irreducible component $H_1$ of $H$ has the form
$H_1=g\Vstar$ for an arbitrary element $g\in H_1$. Moreover, one can choose $g$
such that its coordinates are roots of unity \cite[Remark 3.1, Remark
5.2]{grigoriev_milman2012}.

In addition, in order to obtain the Laurent binomials defining $\Vstar$, the
algorithm finds a $\ZZ$-basis of the intersection
$\QQ(Q_{m+1},\dots,Q_n) \cap \ZZ^n$ of the $\QQ$-linear space generated by
vectors $Q_j=(q_{j1},\dots,q_{jm},0,\dots,0,q_j,0\dots,0)$, $m<j\le n$ with the
grid $\ZZ^n$ \cite[Remark 3.1]{grigoriev_milman2012}. To find that
$\ZZ$-basis, the algorithm first applies
\cite{frumkin_application_1976,dumas_efficient_2001} to produce (within
polynomial complexity) a $\ZZ$-basis $Z$ of the space of integer solutions of
the linear system with rows $Q_j$, $m<j\le n$, and subsequently applies
\cite{frumkin_application_1976,dumas_efficient_2001} to construct a $\ZZ$-basis
of the linear system with the rows from $Z$.

Recall from the Definition~\ref{def:shiftedtoricvariety} that for a torus
$\Vstar\subseteq \Cstarn$ and a point $g\in \Cstarn$ we call $g\Vstar$ a shifted torus,
and $\overline{g\Vstar}\subseteq \CCn$ a shifted toric variety. In particular, every
irreducible component of a binomial variety in $\Cstarn$ is a shifted torus
(Proposition~\ref{structurebinomialvariety}). One can modify the described
algorithm to test whether an input variety $\VV\subseteq \CCn$ is shifted toric. To this
end, pick an arbitrary point $g=(g_1,\dots,g_n)\in \Vstar$ and test whether
$g^{-1}\Vstar$ is a torus. If the latter holds, the algorithm produces binomial
equations for $g^{-1}\Vstar$. Clearly a binomial equation
$x_1^{s_1}\cdots x_n^{s_n}=1$ vanishes on $g^{-1}\Vstar$ if and only if
$x_1^{s_1}\cdots x_n^{s_n}=g_1^{s_1}\cdots g_n^{s_n}$ vanishes on $\Vstar$ .

The following theorem summarizes our algorithm.

\begin{theorem}\label{complex}
  Let $f_1$, \dots, $f_k\in \Zx$, $\deg(f_i)\le d$, $1\le i\le k$ with bit-sizes of integer
  coefficients of $f_1$, \dots,~$f_k$ at most $L$. One can design an algorithm which
  tests whether the variety $\VV\subseteq \CCn$ determined by $f_1$,
  \dots,~$f_k$ is (shifted) toric. In the positive case, the algorithm yields a
  transcendental basis $x_{i_1}$, \dots,~$x_{i_m}$, $m=\dim \VV$ of $\VV$. It
  furthermore yields binomial equations defining $\Vstar=\VV\cap \Cstarn$. Each
  binomial equation has the form
  $x_1^{s_1}\cdots x_n^{s_n}=g_1^{s_1}\cdots g_n^{s_n}$ for all
  $(g_1,\dots,g_n)\in \Vstar$ and for suitable integers $s_1$,
  \dots,~$s_n\in \ZZ$ satisfying $|s_i|<O(d^{n^2})$ . The complexity of the designed
  algorithm does not exceed $(Ld^{n^2})^{O(1)}$.
\end{theorem}

One can extend the algorithm in Theorem~\ref{complex} so that it takes a
reducible variety $\VV$ and decomposes $V$ into irreducible components and then
tests whether its irreducible components are toric or shifted toric (following
the lines of the algorithm for each irreducible component of $\Vstar$
separately). Moreover, if the irreducible components of $\VV$ are exclusively
toric and shifted toric varieties and $\Vstar$ is a group, then the extended
algorithm yields the described representation of $\Vstar$ by means of binomials.
The complexity of the extended algorithm is still $(Ld^{n^2})^{O(1)}$.

Our discussion here can be straightforwardly generalized to any algebraically
closed field with characteristic zero. Independently, the algorithm can be
generalized to coefficients from a finite field extension of $\QQ$
\cite{chistov_algorithm_1986,grigoriev_factorization_1986}.

\subsection{Toricity over Real Closed Fields}\label{subsection:realclosed}
In this subsection we design an algorithm that recognizes toricity of a
semi-algebraic set over $\RR$. We refer to \cite{Basu2016} for the algorithms in
real algebraic geometry. Let $\Rplus:=\{\,z\in \RR\mid z>0\,\}$ denote the positive
orthant. Keeping the notations from Subsection~\ref{subsec:algebraicallyclosed}
consider the semi-algebraic set $T:=\{\,x\in \Rplusn\mid\text{$f_i(x)\ge 0$,
  $1\le i\le k$}\,\}$. Modify (\ref{1}) replacing $\Cstar$ by $\Rplus$ and
equalities $f_i=0$ by inequalities $f_i\ge 0$, respectively. Also keep the
first-order formula (\ref{1}). This formula can be verified by applying the
algorithms from \cite{grigoriev_solving_1988}. Clearly, (\ref{1}) is true if and
only if $T$ is a group (a torus). Thus, assume that (\ref{1}) is true. Then the
image of the coordinate-wise logarithm map $\log (T)\subseteq \RRn$ is a linear
subspace. Hence, in particular, $T$ is connected.

First, compute $m:=\dim (T)$ and produce $m$ coordinates such that the
projection $U$ of $T$ on $m$-dimensional space with these coordinates has the
full dimension $m$. Without loss of generality assume that these coordinates are
$x_1$, \dots,~$x_m$. Fix $j$, $m<j\le n$ and denote by $U_j$ the projection of
$T$ on the $(m+1)$-dimensional space with coordinates $x_1$, \dots,~$x_m$,
$x_j$. Denote by $p$ the projection map of the latter space along $x_j$ onto
$m$-dimensional space with the coordinates $x_1,\dots,x_m$. Then $p(U_j)=U$.
Since $\dim(U_j)=\dim(U)=m$, we have that $p(\log(U_j))=\log(U)=\RR^m$, and
therefore we conclude that any point of $U$ has a unique preimage of $p$ in
$U_j$. Moreover, $U_j$ is determined by a single binomial-type (analytic)
equation of the form
\begin{eqnarray}\label{2}
x_1^{t_{j1}}\cdots x_m^{t_{jm}}  x_j^{t_j}=c_j
\end{eqnarray}
for some reals $t_{j1}$, \dots,~$t_{jm}$, $t_j$, $c_j$. Since $U_j$ is a group
we get that $c_j=1$.

On the other hand, applying the algorithm from \cite{Basu2016}, one can
construct the projection $p:U_j\to U$ and conclude that equation (\ref{2}) is
algebraic, thus $t_{j1}$, \dots,~$t_{jm}$, $t_j \in \ZZ$. The algorithm also
yields $t_{j1}$, \dots,~$t_{jm}$, $t_j$. Note that without loss of generality
one can assume that they are relatively prime, otherwise divide by their
greatest common divisor. The complexity of the algorithm is again
$(Ld^{n^2})^{O(1)}$. Doing so for each $j$, $m<j\le n$, the algorithm yields
binomial equations of the form (\ref{2}) which determine $T$ uniquely. Similar
to subsection \ref{subsec:algebraicallyclosed} one can produce a $\ZZ$-basis of
the intersection of $\QQ$-linear space generated by vectors
$(t_{j1},\dots,t_{jm},0,\dots,0,t_j,0\dots,0)$ with the grid $\ZZ^n$. Then any
vector $(s_1,\dots,s_n)$ from this basis provides a binomial
$X_1^{s_1}\cdots X_n^{s_n}-1$ that vanishes on $T$. The above $\ZZ$-basis need
not be constructed, since binomials of the form (\ref{2}) already determine $T$
uniquely. We summarize the described algorithm in the following theorem.

\begin{theorem}\label{real}
  Let $f_1$, \dots,~$f_k\in \ZZ[x_1,\dots,x_n]$, $\deg(f_i)\le d$, $1\le i\le k$
  with bit-sizes of integer coefficients of $f_1$, \dots,~$f_k$ at most $L$. One
  can design an algorithm which tests whether the semi-algebraic set
  $T:=\{\,x\in \Rplusn\mid\text{$f_i(x)\ge 0$, $1\le i\le k$}\,\}$ is a group (a
  torus). In the positive case, the algorithm yields coordinates $x_{i_1}$,
  \dots,~$x_{i_m}$, where $m=\dim(T)$, such that the dimension of the projection
  of $T$ on the $m$-dimensional space with the coordinates $x_{i_1}$,
  \dots,~$x_{i_m}$ equals $m$. It furthermore yields for each
  $j\notin \{i_1,\dots,i_m\}$ a binomial equation of the form
  $x_1^{t_{j1}}\cdots x_m^{t_{jm}}\cdot x_j^{t_j}=1$ which vanishes on $T$, with
  relatively prime $t_{j1}$, \dots,~$t_{jm}$, $t_j \in \ZZ$, where
  $|t_{j1}|+\cdots+|t_{jm}|+|t_j|\le d^{O(n)}$. The complexity of the algorithm
  does not exceed $(Ld^{n^2})^{O(1)}$.
\end{theorem}

Similar to Subsection \ref{subsec:algebraicallyclosed}, our results here can be
generalized to arbitrary real closed field.

\subsection{Membership in Binomial Varieties}\label{subsection:binomiality}
Let a group $G\subset (\overline{\QQ}^*)^n$ be given by binomial equations
$x_1^{a_{i,1}}\cdots x_n^{a_{i,n}}=1$, $1\le i\le k$, where $a_{i,j}\in \ZZ$,
$|a_{i,j}|\le d$. Let $v=(v_1,\dots,v_n)\in \QQ^n$ be a point such that the
absolute values of the numerators and denominators of $v_1$, \dots,~$v_n$ do not
exceed $M$. We design an algorithm which tests whether $v\in \overline G$.
Recall that $\overline G$ is a binomial variety, and it is a toric variety when
$G$ is irreducible.

\begin{theorem}\label{th:membership}
  There is an algorithm which tests whether a point $v$  belongs to a
  binomial variety $\overline G$ with its complexity bounded by 
  \begin{itemize}
  \item[(i)] $(k\cdot \log M \cdot (dn)^n)^{O(1)}$ and by
  \item[(ii)] $(k\cdot M\cdot n\cdot \log d)^{O(1)}$.
  \end{itemize}
\end{theorem} 

\begin{proof}
  \sloppy Permuting the coordinates, assume w.l.o.g.~that
  $v=(0,\dots,0,v_{s+1},\dots,v_n)$, where $v_{s+1}\cdots v_n \neq 0$. Due to
  Claim 5.3 in \cite{grigoriev_milman2012}, $v\in \overline G$ if and only if
  there exist a point $u=(u_1,\dots,u_s,v_{s+1},\dots,v_n)\in G$ and positive
  $0<b_1,\dots,b_s\in \ZZ$ such that
\begin{eqnarray}\label{3}
v=\lim_{t\to 0}(u_1\cdot t^{b_1},\dots,u_s\cdot t^{b_s},v_{s+1},\dots,v_n)
\end{eqnarray}
and for all $t\neq 0$, i.e., a shift of a one-parametric subgroup, we have that
$(u_1\cdot t^{b_1},\dots,u_s\cdot t^{b_s},v_{s+1},\dots,v_n) \in G$. Then
(\ref{3}) is equivalent to the existence of $u\in G$ and a one-parametric
subgroup $\{\,(t^{b_1},\dots,t^{b_s},1,\dots,1) \mid t\neq 0\,\} \subseteq G$.
The existence of the latter is equivalent to the existence of a non-negative
vector $(b_1,\dots,b_s,0,\dots,0)$ orthogonal to the vectors
$(a_{i,1},\dots,a_{i,n})$, $1\le i\le k$. This can be checked by means of linear
programming.

The existence of a point $u\in G$ satisfying (\ref{3}) is equivalent to the
existence of non-zero $u_1$, \dots,~$u_s\in \overline{\QQ}^*$ satisfying the
binomial equations
\begin{eqnarray}\label{4}
  u_1^{a_{i,1}}\cdots u_s^{a_{i,s}}=v_{s+1}^{-a_{i,s+1}}\cdots v_n^{-a_{i,n}},\quad 1\le i\le k.
\end{eqnarray} 
One can apply the algorithm in \cite{GrigorievWeber2012a} to this system and
transform the $k\times s$ submatrix $A=(a_{i,j})$, $1\le i\le k$, $1\le j\le s$
to its Smith form. Then solvability of (\ref{4}) is equivalent to that the
right-hand side of (\ref{4}) fulfils (at most $k$) relations of the form
\begin{eqnarray}\label{5}
\prod_{1\le i\le k} (v_{s+1}^{-a_{i,s+1}}\cdots v_n^{-a_{i,n}})^{c_i} =1
\end{eqnarray}
for some integers $c_i$ being suitable minors of matrix $A$, hence
$|c_i|\le (ds)^{O(s)}$ by Hadamard's inequality.

One can verify relations (\ref{5}) using the binary form of the numerators and
denominators of $v_{s+1}$, \dots,~$v_n$. This leads to the complexity bound (i).
Alternatively, one can factorize the numerators and denominators of $v_{s+1}$,
\dots,~$v_n$ and execute calculations in terms of exponents of their prime
factors which leads to the complexity bound (ii). This completes the
verification of (\ref{3}) and the description of the algorithm.
\end{proof}

One can extend the complexity bound (i) for $v_j\in \overline{\QQ}$,
$1\le j\le n$ being algebraic numbers. In this case $\log M$ plays the role of
the bit-size of the representation of $v$.

\section{Conclusions and Future Work}\label{sec:conclusions}
We have taken a geometric approach to studying steady state varieties,
which---besides significant theoretical results---generated comprehensive
empirical data from computations on 129 networks from BioModels repository. We
are not aware of any comparable systematic large scale symbolic computations on
those data in the literature. We were indeed surprised by the success rate of
Gröbner basis and real quantifier elimination techniques with input sizes up to
71 variables. We find this most encouraging and believe that robust and
supported software tools for systems biology and medicine that include symbolic
computation components are not out of reach.

It was important to learn that real methods do not significantly fall behind
complex methods efficiency-wise. After all, chemical reaction network theory
takes place in the interior of the positive first orthant. In a way, our
consideration of $V^*$ in favor of $V$ here marks a first step in that direction
by looking at points in the variety rather than polynomials in the ideal. Only
real methods allow to go further and, e.g., identify prime components whose
varieties reach into the first orthant.

Our work here gives a number of quite concrete challenges to be considered in
subsequent work. To start with, one must complete the step from ``coset'' to
``shifted toric'' in practical software by producing suitable prime
decompositions beyond decompositions over $\QQ$. As mentioned above, real
methods allow to explicitly refer to the interior of the first orthant, and our
framework and the first-order descriptions we use should be refined in this
direction. On the complex side, one should also test suitable elimination
methods with the logic descriptions we developed here.

It is noteworthy that all quantifier elimination problems considered here were
decision problems, even without quantifier alternation. On the one hand, this
allows the application of methods and tools from Satisfiability Modulo Theories
Solving \cite{NieuwenhuisEtAl:06a,AbrahamAbbott:2016b}. On the other hand, it
shows that we are not yet using the full power of quantifier elimination
methods. One could, e.g., leave a subset of reaction rates parametric and study
invariance of shifted toricity under variation of those reaction rates.

Our input of 129 models considered here is in a way complete: We took all
available models from BioModels/ODEbase for which we could straightforwardly
produce polynomial vector fields. So far we did not consider systems with
rational vector fields. Such systems come with interesting challenges on the
algebraic side: While the variety $V$ is blind for the presence or non-presence
of polynomials factors from the denominators, those factors can still affect
shifted toricity $V^*$.


\subsection*{Acknowledgments}
This work has been supported by the bilateral project
ANR-17-CE40-0036/DFG-391322026 SYMBIONT
\cite{BoulierFages:18a,BoulierFages:18b}. The first author is grateful to grant
RSF-16-11-10075 and to MCCME for the inspiring atmosphere. Stephen Forrest at
Maplesoft and Arthur Norman at the University of Cambridge greatly helped with
Maple- and Reduce-related technical issues, respectively. Christoph Lüders at
the University of Bonn greatly supported the project with discussions around and
generation of suitable input data.

\appendix
\section{Computations on 129 Models from the BioModels
  Repository}\label{app:a}
We conducted our computations on a 2.40~GHz Intel Xeon E5-4640 with 512 GB RAM
and 32 physical cores providing 64 CPUs via hyper-threading. For
parallelization of the jobs for the individual models we used GNU Parallel
\cite{Tange:11a}. Results are stored in an Sqlite3 database file, which
contains considerably more information than can be presented here in print. It
is available as an ancillary file \texttt{computations.db} with this arXiv
preprint. Beyond our own computations the database imports data from BioModels
and ODEbase, so that our models and the history of our input can be reliably
tracked. For instance, we store the mappings between our variables and the
original species names and even the original SBML file with each model.

Among the information imported from ODEbase there is a binary flag indicating
whether or not a model has mass action kinetics \cite{voit2015150} according
to the following criterion \cite[Section 2.1.2]{feinberg-book}.
\begin{definition}[Mass Action Test]\label{def:ma}
  A system is considered a mass action system when the kinetic law is made up
  of the product of the concentrations of the reactant species to the power of
  their respective stoichiometry times a constant
\end{definition}

\subsection{Classifications of the Original Systems}\label{app:cb}
We start with the classification over $\CC$ and $\RR$ of the original, not
decomposed, systems using algorithms from Subsection~\ref{se:comp-c} and
Subsection~\ref{se:comp-r}, respectively. For comments on the tables from a
theoretical point of view compare those sections.

Beyond the data presented here we generally save the sets $X$, $\hat X$.
Furthermore, over $\CC$ we save the computed Gröbner bases $G$, $\hat G$,
$\tilde G$ and their term orders, and over $\RR$ we save witnesses for
$V_\RR^*\neq\emptyset$.

Recall that over $\CC$, $\hat X$ describes only compatible projection
spaces, while over $\RR$ it describes the canonical projection space.

\begin{remark}\label{rem:heuristic}
  Let $F \subseteq \QQ[X]$. Let $Y_\RR$, $Y_\CC \subseteq X$ describe the canonical projection
  spaces with respect to $V_\RR(F)$, $V_\CC(F)$, respectively. Let furthermore
  $\hat X$ describe any compatible projection space with respect to $V_\CC(F)$.
  Then $Y_\RR \subseteq Y_\CC \subseteq \hat X$ and thus
  $|Y_\RR| \leq |Y_\CC| \leq |\hat X|$. Using the fact that all
  sets are finite, $|\hat X| = |Y_\RR|$ implies $\hat X = Y_\CC$.
\end{remark}

Recall that $\hat X$ as computed in Algorithm~\ref{alg:pacc} describes only a
compatible projection space with respect to the complex variety, while 
$\hat X$ as computed in Algorithm~\ref{alg:pacr} describes the canonical
projection space with respect to the real variety. The idea with
Algorithm~\ref{alg:pacc} was to have a heuristic method to
efficiently obtain a description of the canonical projection space also there.
Remark~\ref{rem:heuristic} tells us that whenever we find equal numbers for
$|\hat X|$ over $\CC$ and $\RR$ in Table~\ref{tab:cbc} below, then that heuristic
method was successful. This is the case with all models where the computation
terminated over both $\CC$ and $\RR$ except for the models 243 and 289.

Since we obtain $\tXX$ in contrast to $\tOO$ with model 289, we know that
$\hat X$ describes a canonical projection space over $\CC$ also there. With
model 243, our obtained $\hat X$ indeed does not describe the canonical
projection space over $\CC$. Inspection of the computation shows that the
Gröbner basis $G$ in l.1 of Algorithm~\ref{alg:pacc} contains $x_6^2$. An
improved heuristic method could check for powers of variables occurring in
$G$. However, Example~\ref{ex:heu} shows that this would still be only
heuristic.

\begin{longtable}{@{\extracolsep{\fill}}|rr|rrrr|rrrr|}
  \captionabove{Applying $\ProjectAndClassify_K$ over $\CC$ and $\RR$. Model
  numbers \textit{nnn} stand for \texttt{BIOMD0000000\itshape{nnn}}. ``m/a''
  indicates mass action kinetics (Definition~\ref{def:ma}). $|X|$ and
  $|\hat X|$ are numbers of variables before and after projection,
  respectively. $\gamma$ is $\tGG$ for group, $\tCC$ fo coset, $\tOO$ for
  empty set, and $\tXX$ else; lower case letters indicate projection. Time
  columns give total CPU times in seconds or ``$\bot$'' for a timeout with a
  limit of 6 hours per model.\label{tab:cbc}}\\
  \hline
  & & \multicolumn{4}{c|}{Algorithm~\ref{alg:pacc} ($\CC$)}
  & \multicolumn{4}{c|}{Algorithm~\ref{alg:pacr} ($\RR$)}
  \\
  model & m/a & $|X|$ & $|\hat X|$
        & $\gamma$ & time (s) & $|X|$ & $|\hat X|$
              & $\gamma$ & time (s)\\
  \hline
  \endfirsthead
  \captionabove{Applying $\ProjectAndClassify_K$ over $\CC$ and $\RR$
  (continued)}\\
  & & \multicolumn{4}{c|}{Algorithm~\ref{alg:pacc} ($\CC$)}
  & \multicolumn{4}{c|}{Algorithm~\ref{alg:pacr} ($\RR$)}
  \\
  model & m/a & $|X|$ & $|\hat X|$
        & $\gamma$ & time (s) & $|X|$ & $|\hat X|$
              & $\gamma$ & time (s)\\
  \hline
  \endhead
  001  &  1  &  12  &  12  & $ \mathtt{C} $ &  10.32  &  12  &  12  & $ \mathtt{C} $ &  3.37\\
  002  &  1  &    &    &    &  $\bot$  &    &    &    &  $\bot$\\
  009  &  0  &  22  &  22  & $ \mathtt{C} $ &  85.96  &  22  &  22  & $ \mathtt{C} $ &  21.26\\
  011  &  0  &  22  &  22  & $ \mathtt{C} $ &  150.21  &  22  &  22  & $ \mathtt{C} $ &  11.20\\
  026  &  1  &  11  &  11  & $ \mathtt{X} $ &  3.99  &  11  &  11  & $ \mathtt{X} $ &  1.12\\
  028  &  1  &  16  &  16  & $ \mathtt{X} $ &  66.60  &    &    &    &  $\bot$\\
  030  &  1  &  18  &  18  & $ \mathtt{X} $ &  57.84  &    &    &    &  $\bot$\\
  035  &  0  &  9  &  9  & $ \mathtt{X} $ &  21.11  &  9  &  9  & $ \mathtt{X} $ &  0.25\\
  038  &  0  &    &    &    &  $\bot$  &    &    &    &  $\bot$\\
  040  &  0  &  3  &  3  & $ \mathtt{X} $ &  4.65  &  3  &  3  & $ \mathtt{X} $ &  0.05\\
  046  &  0  &    &    &    &  $\bot$  &    &    &    &  $\bot$\\
  050  &  0  &  9  &  0  & $ \mathtt{o} $ &  0.09  &  9  &  0  & $ \mathtt{o} $ &  0.02\\
  052  &  0  &  6  &  0  & $ \mathtt{o} $ &  0.09  &  6  &  0  & $ \mathtt{o} $ &  0.01\\
  057  &  0  &  6  &  6  & $ \mathtt{C} $ &  1.53  &  6  &  6  & $ \mathtt{C} $ &  0.11\\
  069  &  0  &  10  &  10  & $ \mathtt{X} $ &  14.09  &    &    &    &  $\bot$\\
  072  &  0  &  7  &  3  & $ \mathtt{o} $ &  0.63  &  7  &  3  & $ \mathtt{o} $ &  0.03\\
  077  &  0  &  7  &  7  & $ \mathtt{C} $ &  5.82  &  7  &  7  & $ \mathtt{C} $ &  0.14\\
  080  &  0  &  10  &  8  & $ \mathtt{o} $ &  3.72  &  10  &  8  & $ \mathtt{o} $ &  0.09\\
  082  &  0  &  10  &  8  & $ \mathtt{o} $ &  2.44  &  10  &  8  & $ \mathtt{o} $ &  0.14\\
  085  &  0  &    &    &    &  $\bot$  &    &    &    &  $\bot$\\
  086  &  0  &    &    &    &  $\bot$  &    &    &    &  $\bot$\\
  091  &  0  &  14  &  0  & $ \mathtt{o} $ &  0.05  &  14  &  0  & $ \mathtt{o} $ &  0.03\\
  092  &  0  &  3  &  3  & $ \mathtt{C} $ &  1.87  &  3  &  3  & $ \mathtt{C} $ &  0.04\\
  099  &  0  &  7  &  7  & $ \mathtt{C} $ &  2.48  &  7  &  7  & $ \mathtt{C} $ &  0.34\\
  101  &  1  &  6  &  6  & $ \mathtt{X} $ &  1.82  &  6  &  6  & $ \mathtt{X} $ &  0.09\\
  102  &  0  &  13  &  13  & $ \mathtt{X} $ &  226.67  &    &    &    &  $\bot$\\
  103  &  0  &  17  &  17  & $ \mathtt{X} $ &  20238.89  &    &    &    &  $\bot$\\
  104  &  0  &  4  &  4  & $ \mathtt{O} $ &  1.36  &  4  &  4  & $ \mathtt{O} $ &  0.01\\
  105  &  0  &  26  &  0  & $ \mathtt{o} $ &  0.24  &  26  &  0  & $ \mathtt{o} $ &  0.51\\
  108  &  0  &    &    &    &  $\bot$  &    &    &    &  $\bot$\\
  122  &  0  &  12  &  12  & $ \mathtt{X} $ &  8639.19  &    &    &    &  $\bot$\\
  123  &  0  &    &    &    &  $\bot$  &    &    &    &  $\bot$\\
  125  &  0  &  5  &  5  & $ \mathtt{X} $ &  2.26  &  5  &  5  & $ \mathtt{X} $ &  0.05\\
  137  &  0  &  21  &  20  & $ \mathtt{x} $ &  50.16  &  21  &  20  & $ \mathtt{x} $ &  2.91\\
  147  &  0  &    &    &    &  $\bot$  &    &    &    &  $\bot$\\
  150  &  0  &  4  &  4  & $ \mathtt{C} $ &  3.43  &  4  &  4  & $ \mathtt{C} $ &  0.04\\
  152  &  0  &    &    &    &  $\bot$  &    &    &    &  $\bot$\\
  153  &  0  &    &    &    &  $\bot$  &    &    &    &  $\bot$\\
  156  &  0  &  3  &  3  & $ \mathtt{C} $ &  0.96  &  3  &  3  & $ \mathtt{C} $ &  0.02\\
  158  &  0  &  3  &  3  & $ \mathtt{X} $ &  1.00  &  3  &  3  & $ \mathtt{X} $ &  0.02\\
  159  &  0  &  3  &  3  & $ \mathtt{C} $ &  2.49  &  3  &  3  & $ \mathtt{C} $ &  0.02\\
  163  &  0  &  16  &  16  & $ \mathtt{X} $ &  9.09  &  16  &  16  & $ \mathtt{X} $ &  1.75\\
  173  &  1  &    &    &    &  $\bot$  &    &    &    &  $\bot$\\
  178  &  0  &  4  &  0  & $ \mathtt{o} $ &  0.06  &  4  &  0  & $ \mathtt{o} $ &  0.00\\
  186  &  1  &  10  &  10  & $ \mathtt{O} $ &  4.63  &  10  &  10  & $ \mathtt{O} $ &  0.14\\
  187  &  1  &  10  &  10  & $ \mathtt{O} $ &  9.85  &  10  &  10  & $ \mathtt{O} $ &  0.12\\
  188  &  1  &  10  &  0  & $ \mathtt{o} $ &  0.10  &  10  &  0  & $ \mathtt{o} $ &  0.02\\
  189  &  1  &  7  &  0  & $ \mathtt{o} $ &  0.03  &  7  &  0  & $ \mathtt{o} $ &  0.01\\
  193  &  1  &  8  &  8  & $ \mathtt{X} $ &  7.11  &  8  &  8  & $ \mathtt{X} $ &  0.13\\
  194  &  1  &  5  &  5  & $ \mathtt{X} $ &  7.07  &  5  &  5  & $ \mathtt{X} $ &  0.03\\
  197  &  0  &  5  &  5  & $ \mathtt{X} $ &  9.93  &  5  &  5  & $ \mathtt{X} $ &  0.15\\
  198  &  1  &  9  &  5  & $ \mathtt{c} $ &  0.79  &  9  &  5  & $ \mathtt{c} $ &  0.06\\
  199  &  0  &  8  &  8  & $ \mathtt{C} $ &  2.79  &  8  &  8  & $ \mathtt{C} $ &  0.18\\
  200  &  0  &    &    &    &  $\bot$  &    &    &    &  $\bot$\\
  205  &  0  &    &    &    &  $\bot$  &    &    &    &  $\bot$\\
  220  &  0  &  56  &  46  & $ \mathtt{o} $ &  5.27  &  56  &  46  & $ \mathtt{o} $ &  88.15\\
  226  &  0  &  14  &  14  & $ \mathtt{X} $ &  14.19  &    &    &    &  $\bot$\\
  227  &  0  &  39  &  0  & $ \mathtt{o} $ &  0.14  &  39  &  0  & $ \mathtt{o} $ &  0.19\\
  229  &  0  &  7  &  7  & $ \mathtt{C} $ &  1.33  &  7  &  7  & $ \mathtt{C} $ &  0.20\\
  230  &  0  &  24  &  23  & $ \mathtt{o} $ &  47.72  &    &    &    &  $\bot$\\
  233  &  0  &  2  &  2  & $ \mathtt{X} $ &  1.03  &  2  &  2  & $ \mathtt{X} $ &  0.01\\
  243  &  0  &  19  &  12  & $ \mathtt{o} $ &  1.32  &  19  &  11  & $ \mathtt{o} $ &  5.21\\
  257  &  1  &  8  &  8  & $ \mathtt{X} $ &  2.65  &    &    &    &  $\bot$\\
  259  &  0  &  16  &  0  & $ \mathtt{o} $ &  0.38  &  16  &  0  & $ \mathtt{o} $ &  0.17\\
  260  &  0  &  16  &  0  & $ \mathtt{o} $ &  0.09  &  16  &  0  & $ \mathtt{o} $ &  0.16\\
  261  &  0  &  16  &  0  & $ \mathtt{o} $ &  0.09  &  16  &  0  & $ \mathtt{o} $ &  0.17\\
  262  &  0  &  9  &  1  & $ \mathtt{c} $ &  3.18  &  9  &  1  & $ \mathtt{c} $ &  0.05\\
  263  &  0  &  9  &  1  & $ \mathtt{c} $ &  1.11  &  9  &  1  & $ \mathtt{c} $ &  0.05\\
  264  &  0  &  11  &  9  & $ \mathtt{c} $ &  11.45  &  11  &  9  & $ \mathtt{c} $ &  0.37\\
  267  &  0  &  3  &  0  & $ \mathtt{o} $ &  0.09  &  3  &  0  & $ \mathtt{o} $ &  0.00\\
  270  &  0  &    &    &    &  $\bot$  &    &    &    &  $\bot$\\
  271  &  1  &  4  &  1  & $ \mathtt{c} $ &  0.20  &  4  &  1  & $ \mathtt{c} $ &  0.01\\
  272  &  1  &  4  &  1  & $ \mathtt{c} $ &  0.27  &  4  &  1  & $ \mathtt{c} $ &  0.01\\
  281  &  0  &  32  &  31  & $ \mathtt{o} $ &  1.20  &  32  &  31  & $ \mathtt{o} $ &  24.50\\
  282  &  1  &  3  &  3  & $ \mathtt{O} $ &  0.16  &  3  &  3  & $ \mathtt{O} $ &  0.01\\
  283  &  1  &  3  &  2  & $ \mathtt{o} $ &  0.28  &  3  &  2  & $ \mathtt{o} $ &  0.01\\
  286  &  0  &    &    &    &  $\bot$  &    &    &    &  $\bot$\\
  287  &  1  &  20  &  20  & $ \mathtt{X} $ &  33.76  &    &    &    &  $\bot$\\
  289  &  0  &  4  &  4  & $ \mathtt{X} $ &  0.84  &  4  &  0  & $ \mathtt{o} $ &  0.01\\
  292  &  0  &  2  &  0  & $ \mathtt{o} $ &  0.16  &  2  &  0  & $ \mathtt{o} $ &  0.01\\
  306  &  0  &  2  &  2  & $ \mathtt{C} $ &  0.25  &  2  &  2  & $ \mathtt{C} $ &  0.01\\
  307  &  0  &  2  &  0  & $ \mathtt{o} $ &  0.14  &  2  &  0  & $ \mathtt{o} $ &  0.01\\
  310  &  0  &  1  &  0  & $ \mathtt{o} $ &  0.02  &  1  &  0  & $ \mathtt{o} $ &  0.00\\
  311  &  0  &  1  &  0  & $ \mathtt{o} $ &  0.02  &  1  &  0  & $ \mathtt{o} $ &  0.00\\
  312  &  0  &  2  &  0  & $ \mathtt{o} $ &  0.02  &  2  &  0  & $ \mathtt{o} $ &  0.00\\
  314  &  0  &  10  &  7  & $ \mathtt{o} $ &  0.73  &  10  &  7  & $ \mathtt{o} $ &  0.07\\
  315  &  1  &    &    &    &  $\bot$  &    &    &    &  $\bot$\\
  321  &  0  &  3  &  0  & $ \mathtt{o} $ &  0.02  &  3  &  0  & $ \mathtt{o} $ &  0.00\\
  332  &  0  &    &    &    &  $\bot$  &  70  &  64  & $ \mathtt{o} $ &  395.97\\
  333  &  0  &  49  &  43  & $ \mathtt{o} $ &  573.10  &  49  &  43  & $ \mathtt{o} $ &  120.98\\
  334  &  0  &    &    &    &  $\bot$  &  69  &  63  & $ \mathtt{o} $ &  271.40\\
  335  &  1  &  29  &  28  & $ \mathtt{o} $ &  97.52  &    &    &    &  $\bot$\\
  344  &  0  &    &    &    &  $\bot$  &    &    &    &  $\bot$\\
  357  &  1  &  8  &  4  & $ \mathtt{o} $ &  0.20  &  8  &  4  & $ \mathtt{o} $ &  0.10\\
  359  &  0  &  8  &  6  & $ \mathtt{o} $ &  0.21  &  8  &  6  & $ \mathtt{o} $ &  0.08\\
  360  &  0  &  8  &  6  & $ \mathtt{o} $ &  0.93  &  8  &  6  & $ \mathtt{o} $ &  0.07\\
  361  &  0  &  8  &  7  & $ \mathtt{o} $ &  1.04  &  8  &  7  & $ \mathtt{o} $ &  0.04\\
  362  &  1  &  29  &  28  & $ \mathtt{o} $ &  1304.71  &  29  &  28  & $ \mathtt{o} $ &  429.42\\
  363  &  0  &  3  &  0  & $ \mathtt{o} $ &  0.07  &  3  &  0  & $ \mathtt{o} $ &  0.00\\
  364  &  1  &  12  &  10  & $ \mathtt{o} $ &  1.03  &  12  &  10  & $ \mathtt{o} $ &  0.54\\
  365  &  1  &  30  &  24  & $ \mathtt{o} $ &  212.01  &    &    &    &  $\bot$\\
  407  &  0  &    &    &    &  $\bot$  &    &    &    &  $\bot$\\
  413  &  1  &  5  &  5  & $ \mathtt{X} $ &  1.05  &  5  &  5  & $ \mathtt{X} $ &  0.06\\
  416  &  0  &  32  &  32  & $ \mathtt{X} $ &  25.87  &  32  &  32  & $ \mathtt{X} $ &  4.17\\
  430  &  0  &    &    &    &  $\bot$  &  23  &  23  & $ \mathtt{X} $ &  13.99\\
  431  &  0  &    &    &    &  $\bot$  &  27  &  27  & $ \mathtt{X} $ &  31.90\\
  439  &  0  &  20  &  20  & $ \mathtt{X} $ &  52.51  &  20  &  20  & $ \mathtt{X} $ &  15.42\\
  459  &  0  &  3  &  3  & $ \mathtt{C} $ &  0.50  &  3  &  3  & $ \mathtt{C} $ &  0.09\\
  460  &  0  &  3  &  3  & $ \mathtt{X} $ &  0.52  &  3  &  3  & $ \mathtt{X} $ &  0.04\\
  475  &  0  &  22  &  20  & $ \mathtt{o} $ &  3.26  &  22  &  20  & $ \mathtt{o} $ &  7.92\\
  478  &  0  &  29  &  28  & $ \mathtt{o} $ &  33.38  &  29  &  28  & $ \mathtt{o} $ &  16.85\\
  479  &  0  &    &    &    &  $\bot$  &    &    &    &  $\bot$\\
  483  &  0  &  6  &  6  & $ \mathtt{X} $ &  1.26  &  6  &  6  & $ \mathtt{X} $ &  0.11\\
  484  &  0  &  1  &  1  & $ \mathtt{C} $ &  0.14  &  1  &  1  & $ \mathtt{C} $ &  0.01\\
  485  &  0  &  1  &  1  & $ \mathtt{X} $ &  1.22  &  1  &  1  & $ \mathtt{X} $ &  0.12\\
  486  &  1  &  2  &  2  & $ \mathtt{C} $ &  0.63  &  2  &  2  & $ \mathtt{C} $ &  0.02\\
  487  &  1  &  6  &  6  & $ \mathtt{C} $ &  0.69  &  6  &  6  & $ \mathtt{C} $ &  0.16\\
  491  &  1  &  57  &  57  & $ \mathtt{G} $ &  12.21  &    &    &    &  $\bot$\\
  492  &  1  &  52  &  52  & $ \mathtt{G} $ &  17.96  &    &    &    &  $\bot$\\
  504  &  0  &    &    &    &  $\bot$  &    &    &    &  $\bot$\\
  519  &  0  &  3  &  3  & $ \mathtt{C} $ &  4.71  &  3  &  3  & $ \mathtt{C} $ &  0.18\\
  546  &  0  &  3  &  0  & $ \mathtt{o} $ &  0.07  &  3  &  0  & $ \mathtt{o} $ &  0.00\\
  559  &  0  &  71  &  0  & $ \mathtt{o} $ &  0.32  &  71  &  0  & $ \mathtt{o} $ &  0.91\\
  581  &  0  &    &    &    &  $\bot$  &  25  &  25  & $ \mathtt{X} $ &  6.00\\
  584  &  0  &  9  &  9  & $ \mathtt{C} $ &  0.60  &  9  &  9  & $ \mathtt{C} $ &  0.13\\
  619  &  1  &  8  &  0  & $ \mathtt{o} $ &  0.07  &  8  &  0  & $ \mathtt{o} $ &  0.01\\
  629  &  0  &  5  &  5  & $ \mathtt{C} $ &  0.41  &  5  &  5  & $ \mathtt{C} $ &  0.08\\
  637  &  1  &  12  &  12  & $ \mathtt{X} $ &  492.76  &    &    &    &  $\bot$\\
  647  &  1  &  11  &  11  & $ \mathtt{X} $ &  6.03  &  11  &  11  & $ \mathtt{X} $ &  0.37\\
  \hline
\end{longtable}

There are 15 models where the complex classification in Table~\ref{tab:cbc}
succeeded but the real classification timed out: 028, 030, 069, 102, 103, 122,
226, 230, 257, 287, 335, 365, 491, 492, 637. Among those, models 491 and 492
have classification $\tGG$ and models 230, 335, and 265 have classification
$\too$ over $\CC$, from which we can conclude that they have the same
classification over $\RR$, respectively.
Vice versa, there are 5 models where real classification succeeded but complex
classification timed out: 332, 334, 430, 431, 581.
%
There is one single model where we succeeded over both $\CC$ and $\RR$ but
obtained different classifications: model 289 has $\tXX$ over $\CC$ but $\too$
over $\RR$.

Table~\ref{tab:statistics-b} collects some statistical information about the
computations. Figure~\ref{fig:analysis-b} provides some analysis of the
computation times. Notice that many computations finish quite quickly.

\begin{table}[!h]
  \centering
  \captionabove{Statistical information about the computations in
    Table~\ref{tab:cbc}.\label{tab:statistics-b}}
  \begin{tabular}{lrr}
    \hline
    & Algorithm~\ref{alg:pacc} ($\CC$) & Algorithm~\ref{alg:pacr} ($\CC$)\\
    \hline
    time limit & 6\,h & 6\,h\\
    \#\,models & 129 & 129\\
    \#\,successful computations  & 104 & 94\\
    success rate & 80.62\% & 72.87\%\\
    median(time) & 1.33\,s & 0.09\,s\\
    \hline
  \end{tabular}
\end{table}

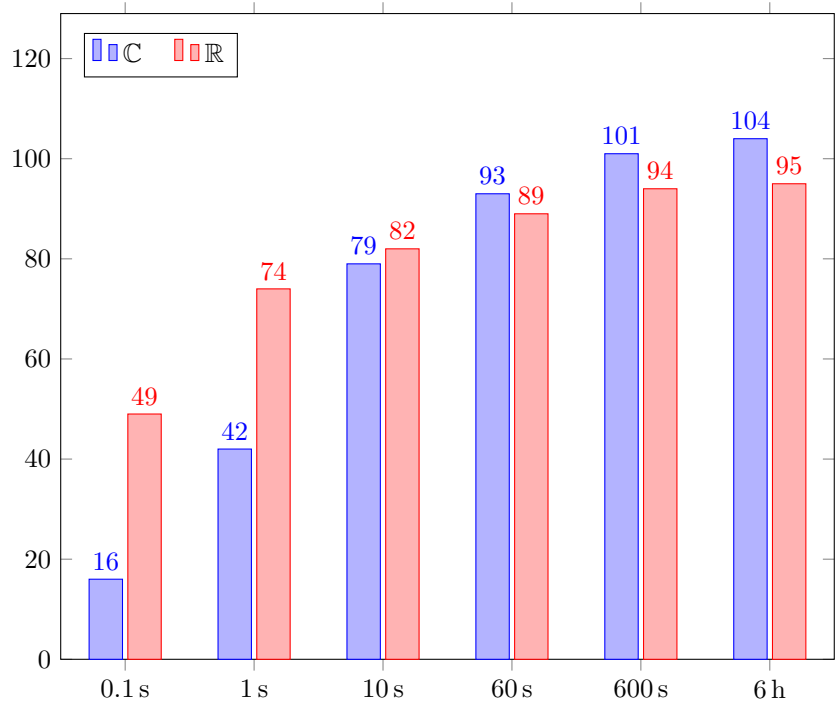
\begin{figure}[!h]
  \centering
  \pgfplotsset{width=0.8\linewidth, compat=1.14}
  \begin{tikzpicture}
    \begin{axis}[
      ymin=0,
      ymax=129,
      ybar,
      bar width=1.25em,
      legend pos=north west,
      legend style={legend columns=2},
      symbolic x coords={0.1\,s, 1\,s, 10\,s, 60\,s, 600\,s, 6\,h},
      xtick=data,
      nodes near coords,
      nodes near coords align={vertical}
      ]
      \addplot coordinates {
        (0.1\,s, 16) (1\,s, 42) (10\,s,79)
        (60\,s,93) (600\,s,101)
        (6\,h, 104)
      };
      \addplot coordinates {
        (0.1\,s, 49) (1\,s, 74) (10\,s,82)
        (60\,s,89) (600\,s,94)
        (6\,h, 95)
      };
      \legend{$\CC\quad$, $\RR$}
    \end{axis}
  \end{tikzpicture}
  \caption{Numbers of problems solved within certain time limits by
    Algorithm~\ref{alg:pacc} over $\CC$ (left) and Algorithm~\ref{alg:pacr}
    over $\RR$ (right). The total number of problems is
    129.\label{fig:analysis-b}}
\end{figure}

\subsection{Classifications of Rational Prime Decompositions}\label{app:cp}
Recall that Algorithm~\ref{alg:dacc} and Algorithm~\ref{alg:dacr} compute
prime decompositions $\mathcal{P}$ over $\QQ$ and then apply our
classification approach to each prime component individually. This yields
lists $\mathcal{X}$, $\mathcal{\hat X}$ containing in turn lists $X$ and
$\hat X$ of variables before and after projection, respectively, as well as a
list $\Gamma$ of classifications $\gamma$. We have
$|\mathcal{X}|=|\mathcal{\hat X}|=|\Gamma|=|\mathcal{P}|$, and elements can be
matched by position. Since this information is too comprehensive to be
displayed in a table, we give only $|\mathcal{P}|$ and summarize the numbers
of occurrences of the various classifications in $\Gamma$. Our database, of
course, stores the complete information.

\begingroup\setlength{\tabcolsep}{5pt}
\begin{longtable}{@{\extracolsep{\fill}}|rr|rr|rr|rr|}
  \captionabove{Applying $\DecomposeAndClassify_K$ over $\CC$ and $\RR$. Model
  numbers \textit{nnn} stand for \texttt{BIOMD0000000\itshape{nnn}}. ``m/a''
  indicates mass action kinetics (Definition~\ref{def:ma}). $|\mathcal{P}|$ is
  the number of prime components over $\QQ$. $\Gamma_\text{summary}$
  summarizes the classification of the components using $\tGG$ for group,
  $\tCC$ fo coset, $\tOO$ for empty set, and $\tXX$ else; lower case letters
  indicate projection. Time columns give total CPU times in seconds or
  ``$\bot$'' for a timeout with a limit of 6 hours per
  model. \label{tab:pc}}\\
  \hline
  &
  &
  &
  & \multicolumn{2}{c|}{Algorithm~\ref{alg:dacc} ($\CC$)}
  & \multicolumn{2}{c|}{Algorithm~\ref{alg:dacr} ($\RR$)}\\
  model & m/a
  & $|\mathcal{P}|$ & time (s)
  & $\Gamma_\text{summary}$ & time (s) & $\Gamma_\text{summary}$ & time (s)\\
  \hline
  \endfirsthead
  \captionabove{Applying $\DecomposeAndClassify_K$ over $\CC$ and $\RR$
  (continued)}\\
  &
  &
  &
  & \multicolumn{2}{c|}{Algorithm~\ref{alg:dacc} ($\CC$)}
  & \multicolumn{2}{c|}{Algorithm~\ref{alg:dacr} ($\RR$)}\\
  model & m/a
  & $|\mathcal{P}|$ & time (s)
  & $\Gamma_\text{summary}$ & time (s) & $\Gamma_\text{summary}$ & time (s)\\
  \hline
  \endhead
  001  &  1  &  1  &  4.49  & $ \mathtt{C} $ &  12.64  & $ \mathtt{C} $ &  5.00\\
  002  &  1  &  1  &  192.42  & $ \mathtt{X} $ &  318.37  &    &  $\bot$\\
  009  &  0  &  28  &  199.03  & $ \mathtt{C} + 13\mathtt{c} + 14\mathtt{o} $ &  238.71  & $ \mathtt{C} + 13\mathtt{c} + 14\mathtt{o} $ &  255.23\\
  011  &  0  &  20  &  167.45  & $ \mathtt{C} + 9\mathtt{c} + 10\mathtt{o} $ &  197.90  & $ \mathtt{C} + 9\mathtt{c} + 10\mathtt{o} $ &  212.07\\
  026  &  1  &  2  &  4.74  & $ \mathtt{o} + \mathtt{X} $ &  5.95  & $ \mathtt{o} + \mathtt{X} $ &  5.36\\
  028  &  1  &  2  &  135.05  & $ \mathtt{o} + \mathtt{X} $ &  162.90  &    &  $\bot$\\
  030  &  1  &  2  &  157.89  & $ \mathtt{o} + \mathtt{X} $ &  175.57  &    &  $\bot$\\
  035  &  0  &  1  &  12.31  & $ \mathtt{X} $ &  18.26  & $ \mathtt{X} $ &  12.76\\
  038  &  0  &  1  &  391.05  & $ \mathtt{X} $ &  463.38  &    &  $\bot$\\
  040  &  0  &  2  &  1.64  & $ \mathtt{o} + \mathtt{X} $ &  2.14  & $ \mathtt{o} + \mathtt{X} $ &  1.67\\
  046  &  0  &  2  &  764.98  & $ \mathtt{X} + \mathtt{x} $ &  817.90  &    &  $\bot$\\
  050  &  0  &  1  &  0.59  & $ \mathtt{o} $ &  0.61  & $ \mathtt{o} $ &  0.60\\
  052  &  0  &  1  &  0.95  & $ \mathtt{o} $ &  0.98  & $ \mathtt{o} $ &  0.95\\
  057  &  0  &  1  &  0.63  & $ \mathtt{C} $ &  5.44  & $ \mathtt{C} $ &  0.70\\
  069  &  0  &  2  &  56.34  & $ 2\mathtt{x} $ &  61.34  &    &  $\bot$\\
  072  &  0  &  2  &  0.26  & $ 2\mathtt{c} $ &  0.55  & $ 2\mathtt{c} $ &  0.27\\
  077  &  0  &  1  &  1.05  & $ \mathtt{C} $ &  6.26  & $ \mathtt{C} $ &  1.15\\
  080  &  0  &  7  &  9.43  & $ 3\mathtt{c} + 4\mathtt{o} $ &  11.43  & $ 3\mathtt{c} + 4\mathtt{o} $ &  9.62\\
  082  &  0  &  7  &  8.45  & $ 3\mathtt{c} + 4\mathtt{o} $ &  13.70  & $ 3\mathtt{c} + 4\mathtt{o} $ &  8.64\\
  085  &  0  &    &  $\bot$  &    &  $\bot$  &    &  $\bot$\\
  086  &  0  &    &  $\bot$  &    &  $\bot$  &    &  $\bot$\\
  091  &  0  &  1  &  0.02  & $ \mathtt{o} $ &  0.06  & $ \mathtt{O} $ &  0.02\\
  092  &  0  &  2  &  0.75  & $ \mathtt{C} + \mathtt{o} $ &  1.36  & $ \mathtt{C} + \mathtt{o} $ &  0.77\\
  099  &  0  &  1  &  0.26  & $ \mathtt{C} $ &  0.92  & $ \mathtt{C} $ &  0.47\\
  101  &  1  &  1  &  0.36  & $ \mathtt{X} $ &  5.90  & $ \mathtt{X} $ &  0.43\\
  102  &  0  &  1  &  38.46  & $ \mathtt{X} $ &  249.72  &    &  $\bot$\\
  103  &  0  &  1  &  3090.38  & $ \mathtt{X} $ &  20628.00  &    &  $\bot$\\
  104  &  0  &  4  &  0.63  & $ 4\mathtt{o} $ &  0.64  & $ 4\mathtt{o} $ &  0.63\\
  105  &  0  &  1  &  0.37  & $ \mathtt{o} $ &  0.55  & $ \mathtt{O} $ &  0.37\\
  108  &  0  &  2  &  40.86  & $ 2\mathtt{X} $ &  68.74  &    &  $\bot$\\
  122  &  0  &  1  &  326.19  & $ \mathtt{X} $ &  363.45  &    &  $\bot$\\
  123  &  0  &    &  $\bot$  &    &  $\bot$  &    &  $\bot$\\
  125  &  0  &  1  &  0.18  & $ \mathtt{X} $ &  1.22  & $ \mathtt{X} $ &  0.21\\
  137  &  0  &  5  &  180.68  & $ 5\mathtt{x} $ &  220.13  & $ 5\mathtt{x} $ &  190.34\\
  147  &  0  &    &  $\bot$  &    &  $\bot$  &    &  $\bot$\\
  150  &  0  &  1  &  1.12  & $ \mathtt{C} $ &  2.86  & $ \mathtt{C} $ &  1.15\\
  152  &  0  &    &  $\bot$  &    &  $\bot$  &    &  $\bot$\\
  153  &  0  &    &  $\bot$  &    &  $\bot$  &    &  $\bot$\\
  156  &  0  &  2  &  0.53  & $ \mathtt{C} + \mathtt{o} $ &  0.77  & $ \mathtt{C} + \mathtt{o} $ &  0.54\\
  158  &  0  &  1  &  0.22  & $ \mathtt{X} $ &  1.06  & $ \mathtt{X} $ &  0.24\\
  159  &  0  &  1  &  0.20  & $ \mathtt{C} $ &  0.60  & $ \mathtt{C} $ &  0.23\\
  163  &  0  &  1  &  9.69  & $ \mathtt{X} $ &  19.50  & $ \mathtt{X} $ &  10.36\\
  173  &  1  &  1  &  295.07  & $ \mathtt{x} $ &  1719.63  &    &  $\bot$\\
  178  &  0  &  1  &  3.66  & $ \mathtt{o} $ &  3.67  & $ \mathtt{o} $ &  3.66\\
  186  &  1  &  5  &  4.73  & $ 2\mathtt{c} + 3\mathtt{x} $ &  7.78  & $ 2\mathtt{c} + 3\mathtt{x} $ &  5.09\\
  187  &  1  &  5  &  18.41  & $ 2\mathtt{c} + 3\mathtt{x} $ &  29.25  & $ 2\mathtt{c} + 3\mathtt{x} $ &  18.79\\
  188  &  1  &  1  &  0.27  & $ \mathtt{o} $ &  0.29  & $ \mathtt{O} $ &  0.27\\
  189  &  1  &  1  &  0.02  & $ \mathtt{o} $ &  0.04  & $ \mathtt{O} $ &  0.02\\
  193  &  1  &  1  &  6.50  & $ \mathtt{X} $ &  7.12  & $ \mathtt{X} $ &  6.62\\
  194  &  1  &  1  &  1.38  & $ \mathtt{X} $ &  1.98  & $ \mathtt{X} $ &  1.41\\
  197  &  0  &  1  &  4.18  & $ \mathtt{X} $ &  6.96  & $ \mathtt{X} $ &  4.24\\
  198  &  1  &  1  &  0.69  & $ \mathtt{c} $ &  3.16  & $ \mathtt{c} $ &  0.72\\
  199  &  0  &  1  &  0.35  & $ \mathtt{C} $ &  2.63  & $ \mathtt{C} $ &  0.50\\
  200  &  0  &    &  $\bot$  &    &  $\bot$  &    &  $\bot$\\
  205  &  0  &    &  $\bot$  &    &  $\bot$  &    &  $\bot$\\
  220  &  0  &    &  $\bot$  &    &  $\bot$  &    &  $\bot$\\
  226  &  0  &  1  &  17.10  & $ \mathtt{X} $ &  25.93  &    &  $\bot$\\
  227  &  0  &  1  &  0.03  & $ \mathtt{o} $ &  0.15  & $ \mathtt{O} $ &  0.03\\
  229  &  0  &  2  &  1.60  & $ \mathtt{C} + \mathtt{c} $ &  6.77  & $ \mathtt{C} + \mathtt{c} $ &  1.73\\
  230  &  0  &  5  &  181.14  & $ 5\mathtt{x} $ &  249.78  &    &  $\bot$\\
  233  &  0  &  3  &  0.14  & $ 2\mathtt{C} + \mathtt{o} $ &  2.46  & $ 2\mathtt{C} + \mathtt{o} $ &  0.15\\
  243  &  0  &  8  &  9.13  & $ 8\mathtt{o} $ &  9.19  & $ 8\mathtt{o} $ &  9.18\\
  257  &  1  &  2  &  0.63  & $ \mathtt{o} + \mathtt{X} $ &  2.29  & $ \mathtt{o} + \mathtt{X} $ &  590.61\\
  259  &  0  &  1  &  2.93  & $ \mathtt{o} $ &  2.98  & $ \mathtt{o} $ &  2.95\\
  260  &  0  &  1  &  1.41  & $ \mathtt{o} $ &  1.46  & $ \mathtt{o} $ &  1.42\\
  261  &  0  &  1  &  0.84  & $ \mathtt{o} $ &  0.87  & $ \mathtt{o} $ &  0.86\\
  262  &  0  &  1  &  1.83  & $ \mathtt{c} $ &  4.43  & $ \mathtt{c} $ &  1.85\\
  263  &  0  &  1  &  2.48  & $ \mathtt{c} $ &  2.66  & $ \mathtt{c} $ &  2.49\\
  264  &  0  &  1  &  3.12  & $ \mathtt{c} $ &  7.28  & $ \mathtt{c} $ &  3.31\\
  267  &  0  &  1  &  0.10  & $ \mathtt{o} $ &  0.11  & $ \mathtt{o} $ &  0.10\\
  270  &  0  &  1  &  14511.34  & $ \mathtt{X} $ &  15819.74  &    &  $\bot$\\
  271  &  1  &  1  &  3.87  & $ \mathtt{c} $ &  5.31  & $ \mathtt{c} $ &  3.88\\
  272  &  1  &  1  &  0.12  & $ \mathtt{c} $ &  1.11  & $ \mathtt{c} $ &  0.12\\
  281  &  0  &  3144  &  8264.69  & $ 1008\mathtt{c} + 2136\mathtt{o} $ &  8317.80  & $ 1008\mathtt{c} + 2136\mathtt{o} $ &  8655.75\\
  282  &  1  &  2  &  0.36  & $ 2\mathtt{o} $ &  0.37  & $ 2\mathtt{o} $ &  0.37\\
  283  &  1  &  2  &  0.10  & $ 2\mathtt{o} $ &  0.11  & $ 2\mathtt{o} $ &  0.11\\
  286  &  0  &    &  $\bot$  &    &  $\bot$  &    &  $\bot$\\
  287  &  1  &  1  &  12.78  & $ \mathtt{X} $ &  27.29  &    &  $\bot$\\
  289  &  0  &  2  &  1.15  & $ \mathtt{o} + \mathtt{X} $ &  2.84  & $ 2\mathtt{o} $ &  1.15\\
  292  &  0  &  1  &  1.20  & $ \mathtt{o} $ &  1.20  & $ \mathtt{O} $ &  1.20\\
  306  &  0  &  2  &  0.53  & $ \mathtt{C} + \mathtt{o} $ &  1.04  & $ \mathtt{C} + \mathtt{o} $ &  0.54\\
  307  &  0  &  1  &  0.02  & $ \mathtt{o} $ &  0.03  & $ \mathtt{o} $ &  0.02\\
  310  &  0  &  1  &  0.02  & $ \mathtt{o} $ &  0.03  & $ \mathtt{o} $ &  0.02\\
  311  &  0  &  1  &  0.02  & $ \mathtt{o} $ &  0.02  & $ \mathtt{o} $ &  0.02\\
  312  &  0  &  1  &  0.04  & $ \mathtt{o} $ &  0.05  & $ \mathtt{o} $ &  0.04\\
  314  &  0  &  3  &  0.94  & $ 3\mathtt{c} $ &  1.11  & $ 3\mathtt{c} $ &  0.99\\
  315  &  1  &    &  $\bot$  &    &  $\bot$  &    &  $\bot$\\
  321  &  0  &  1  &  0.07  & $ \mathtt{o} $ &  0.08  & $ \mathtt{o} $ &  0.07\\
  332  &  0  &    &  $\bot$  &    &  $\bot$  &    &  $\bot$\\
  333  &  0  &    &  $\bot$  &    &  $\bot$  &    &  $\bot$\\
  334  &  0  &    &  $\bot$  &    &  $\bot$  &    &  $\bot$\\
  335  &  1  &    &  $\bot$  &    &  $\bot$  &    &  $\bot$\\
  344  &  0  &    &  $\bot$  &    &  $\bot$  &    &  $\bot$\\
  357  &  1  &  2  &  0.23  & $ 2\mathtt{o} $ &  0.25  & $ 2\mathtt{o} $ &  0.24\\
  359  &  0  &  5  &  0.96  & $ 2\mathtt{c} + 3\mathtt{o} $ &  2.80  & $ 2\mathtt{c} + 3\mathtt{o} $ &  1.02\\
  360  &  0  &  4  &  0.72  & $ 2\mathtt{c} + 2\mathtt{o} $ &  1.48  & $ 2\mathtt{c} + 2\mathtt{o} $ &  0.77\\
  361  &  0  &  2  &  0.48  & $ 2\mathtt{c} $ &  1.01  & $ 2\mathtt{c} $ &  0.60\\
  362  &  1  &    &  $\bot$  &    &  $\bot$  &    &  $\bot$\\
  363  &  0  &  1  &  0.26  & $ \mathtt{o} $ &  0.27  & $ \mathtt{o} $ &  0.26\\
  364  &  1  &  4  &  2.23  & $ 2\mathtt{c} + 2\mathtt{o} $ &  3.22  & $ 2\mathtt{c} + 2\mathtt{o} $ &  2.37\\
  365  &  1  &    &  $\bot$  &    &  $\bot$  &    &  $\bot$\\
  407  &  0  &    &  $\bot$  &    &  $\bot$  &    &  $\bot$\\
  413  &  1  &  1  &  1.07  & $ \mathtt{X} $ &  1.59  & $ \mathtt{X} $ &  1.11\\
  416  &  0  &  3  &  36.24  & $ \mathtt{X} + 2\mathtt{x} $ &  43.95  & $ \mathtt{X} + 2\mathtt{x} $ &  41.76\\
  430  &  0  &  20  &  112.47  & $ 8\mathtt{c} + 10\mathtt{o} + \mathtt{X} + \mathtt{x} $ &  221.05  & $ 8\mathtt{c} + 10\mathtt{o} + \mathtt{X} + \mathtt{x} $ &  12709.47\\
  431  &  0  &    &  $\bot$  &    &  $\bot$  &    &  $\bot$\\
  439  &  0  &  9  &  32.38  & $ 2\mathtt{c} + 2\mathtt{o} + \mathtt{X} + 4\mathtt{x} $ &  40.10  &    &  $\bot$\\
  459  &  0  &  1  &  0.24  & $ \mathtt{C} $ &  0.50  & $ \mathtt{C} $ &  0.25\\
  460  &  0  &  1  &  1.27  & $ \mathtt{X} $ &  1.49  & $ \mathtt{X} $ &  1.30\\
  475  &  0  &  30  &  30.46  & $ 14\mathtt{c} + 4\mathtt{o} + 12\mathtt{x} $ &  37.62  & $ 14\mathtt{c} + 4\mathtt{o} + 12\mathtt{x} $ &  43.66\\
  478  &  0  &    &  $\bot$  &    &  $\bot$  &    &  $\bot$\\
  479  &  0  &    &  $\bot$  &    &  $\bot$  &    &  $\bot$\\
  483  &  0  &  1  &  0.23  & $ \mathtt{X} $ &  0.59  & $ \mathtt{X} $ &  0.28\\
  484  &  0  &  1  &  0.10  & $ \mathtt{C} $ &  0.14  & $ \mathtt{C} $ &  0.11\\
  485  &  0  &  1  &  0.16  & $ \mathtt{X} $ &  0.23  & $ \mathtt{X} $ &  0.23\\
  486  &  1  &  1  &  0.15  & $ \mathtt{C} $ &  0.27  & $ \mathtt{C} $ &  0.16\\
  487  &  1  &  1  &  1.40  & $ \mathtt{C} $ &  1.78  & $ \mathtt{C} $ &  1.49\\
  491  &  1  &  1  &  2.55  & $ \mathtt{G} $ &  22.49  & $ \mathtt{G} $ &  47.84\\
  492  &  1  &  1  &  1.59  & $ \mathtt{G} $ &  10.04  & $ \mathtt{G} $ &  35.49\\
  504  &  0  &    &  $\bot$  &    &  $\bot$  &    &  $\bot$\\
  519  &  0  &  3  &  1.60  & $ \mathtt{C} + \mathtt{c} + \mathtt{o} $ &  4.15  & $ \mathtt{C} + \mathtt{c} + \mathtt{o} $ &  1.62\\
  546  &  0  &  1  &  0.13  & $ \mathtt{o} $ &  0.14  & $ \mathtt{o} $ &  0.13\\
  559  &  0  &  1  &  0.98  & $ \mathtt{o} $ &  1.18  & $ \mathtt{O} $ &  0.99\\
  581  &  0  &    &  $\bot$  &    &  $\bot$  &    &  $\bot$\\
  584  &  0  &  1  &  0.26  & $ \mathtt{C} $ &  2.36  & $ \mathtt{C} $ &  0.38\\
  619  &  1  &  1  &  0.86  & $ \mathtt{o} $ &  0.88  & $ \mathtt{o} $ &  0.88\\
  629  &  0  &  1  &  1.03  & $ \mathtt{C} $ &  1.70  & $ \mathtt{C} $ &  1.06\\
  637  &  1  &  3  &  8383.02  & $ \mathtt{X} + 2\mathtt{x} $ &  8464.64  &    &  $\bot$\\
  647  &  1  &  1  &  5.95  & $ \mathtt{X} $ &  6.70  & $ \mathtt{X} $ &  6.23\\
  \hline
\end{longtable}
\endgroup

There are 17 models where the complex classification in Table~\ref{tab:pc}
succeeded but the real classification timed out: 002, 028, 030, 038, 046, 069,
102, 103, 108, 122, 173, 226, 230, 270, 287, 439, 637. Vice versa, there are
no models where the classification succeeded over $\RR$ but not over $\CC$.
There are 8 models where we succeeded over both $\CC$ and $\RR$ but obtained
different classifications: 091, 105, 188, 189, 227, 289, 292, 559. All those
differences are visible in the summaries $\Gamma_{\text{summary}}$ in
Table~\ref{tab:pc}. Model 289 is $\too+\tXX$ over $\CC$ but $2\too$ over
$\RR$. We have addressed this difference already with the computations in
Appendix~\ref{app:cb}. With all other models listed above the difference is
$\too$ over $\CC$ in contrast to $\tOO$ over $\RR$.

Table~\ref{tab:statistics-p} collects some statistical information about the
computations. Figure~\ref{fig:analysis-p} provides some analysis of the
computation times. Notice that many computations finish quite quickly.

\begin{table}[!h]
  \centering
  \captionabove{Statistical information about the computations in
    Table~\ref{tab:pc}.\label{tab:statistics-p}}
  \begin{tabular}{lrr}
    \hline
    & Algorithm~\ref{alg:dacc} ($\CC$) & Algorithm~\ref{alg:dacr} ($\CC$)\\
    \hline
    time limit & 6\,h & 6\,h\\
    \#\,models & 129 & 129\\
    \#\,successful computations  & 105 & 88\\
    success rate & 81.40\% & 68.22\%\\
    median(time) & 2.80\,sec & 0.99\,sec\\
    \hline
  \end{tabular}
\end{table}

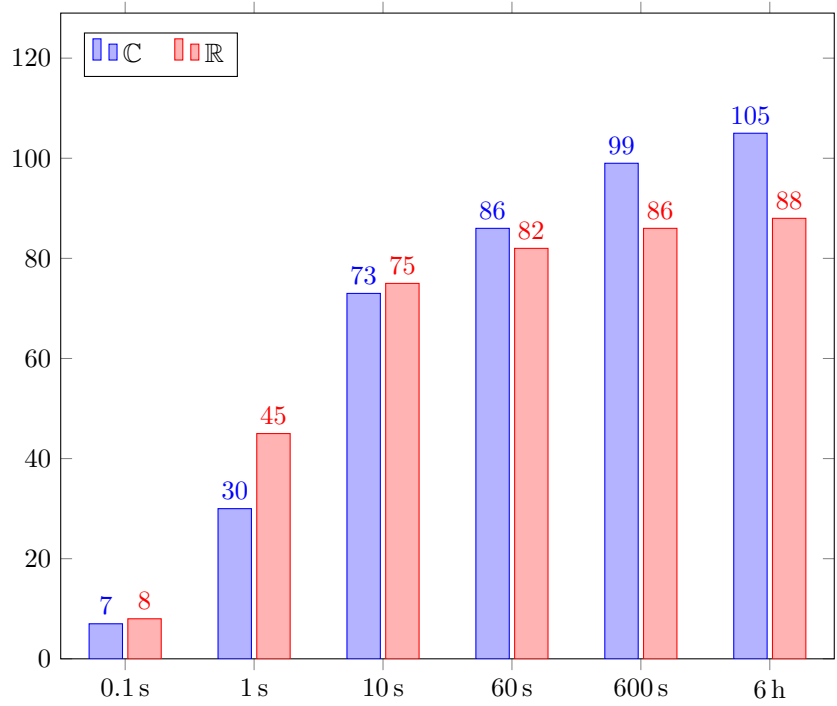
\begin{figure}[!h]
  \centering
  \pgfplotsset{width=0.8\linewidth, compat=1.14}
  \begin{tikzpicture}
    \begin{axis}[
      ymin=0,
      ymax=129,
      ybar,
      bar width=1.25em,
      legend pos=north west,
      legend style={legend columns=2},
      symbolic x coords={0.1\,s, 1\,s, 10\,s, 60\,s, 600\,s, 6\,h},
      xtick=data,
      nodes near coords,
      nodes near coords align={vertical}
      ]
      \addplot coordinates {
        (0.1\,s, 7) (1\,s, 30) (10\,s,73)
        (60\,s,86) (600\,s,99)
        (6\,h, 105)
      };
      \addplot coordinates {
        (0.1\,s, 8) (1\,s, 45) (10\,s,75)
        (60\,s,82) (600\,s,86)
        (6\,h, 88)
      };
      \legend{$\CC\quad$, $\RR$}
    \end{axis}
  \end{tikzpicture}
  \caption{Numbers of problems solved within certain time limits by
    Algorithm~\ref{alg:dacc} over $\CC$ (left) and Algorithm~\ref{alg:dacr}
    over $\RR$ (right). The total number of problems is
    129.\label{fig:analysis-p}}
\end{figure}
\end{document}